\documentclass[12pt]{iopart}
\usepackage{iopams}
\usepackage{amsfonts}
\usepackage{amsthm}
\usepackage{amsbsy}
\usepackage{amssymb}
\usepackage{times,graphicx,xcolor}
\usepackage{delarray,fancybox}
\usepackage{mathdots}
\usepackage{color}
\usepackage{url}
\usepackage{bm}
\usepackage[breaklinks,linktocpage]{hyperref}     
\usepackage{subfigure}
\expandafter\let\csname equation*\endcsname\relax
\expandafter\let\csname endequation*\endcsname\relax
\expandafter\let\csname leftroot\endcsname\relax
\usepackage{amsmath}
\newcommand{\ito}[1]{\overset{\scriptscriptstyle{0}}{#1}}
\newcommand{\str}[1]{\overset{\scriptscriptstyle{1/2}}{#1}}
\newcommand{\oti}[1]{\overset{\scriptscriptstyle{1}}{#1}}
\newcommand{\dt}{\mathrm{d}t}
\newcommand{\ti}{t_{\mathrm{o}}}

\newcommand{\tf}{t_{\mathrm{f}}}

\newcommand{\id}{\mathsf{1}}
\newcommand{\hoz}{\tf,\ti}
\newcommand{\brl}{\bm{\left(}}
\newcommand{\brr}{\bm{\right)}}
\theoremstyle{plain} 
\newtheorem{proposition}{Proposition}[section]

\theoremstyle{definition}

\theoremstyle{remark}

\begin{document}

\title[On extremals of the entropy production by ``Langevin--Kramers'' dynamics]
{On extremals of the entropy production by ``Langevin--Kramers'' dynamics}

\author{Paolo Muratore-Ginanneschi}
\address{Department of Mathematics and Statistics
PL 68
FIN-00014 University of Helsinki
Finland}

\ead{paolo.muratore-ginanneschi@helsinki.fi}


\begin{abstract} 
We refer as ``Langevin--Kramers'' dynamics to a class of stochastic differential systems exhibiting 
a degenerate ``metriplectic'' structure. This means that the drift field can be decomposed into 
a symplectic and a gradient-like component with respect to a pseudo-metric tensor 
associated to random fluctuations affecting increments of only a sub-set of the
degrees of freedom. Systems in this class are often encountered in applications 
as elementary models of Hamiltonian dynamics in an heat bath eventually relaxing to a 
Boltzmann steady state. 

Entropy production control in Langevin--Kramers models differs from the now well-understood 
case of Langevin--Smoluchowski dynamics for two reasons. 
First, the definition of entropy production stemming from fluctuation theorems specifies 
a cost functional which does not act \emph{coercively} on all degrees of freedom 
of control protocols.
Second, the presence of a symplectic structure imposes a \emph{non-local} constraint 
on the class of admissible controls.
Using Pontryagin control theory and restricting the attention to additive 
noise, we show that \emph{smooth} protocols attaining extremal values of the entropy 
production appear generically in continuous parametric families as a consequence 
of a trade-off between smoothness of the admissible protocols and non-coercivity 
of the cost functional. Uniqueness is, however, always recovered in the over-damped limit as 
extremal equations reduce at leading order to the 
Monge--Amp\`ere--Kantorovich optimal mass transport equations.
\end{abstract}

\maketitle

\section{Introduction}
\label{sec:intro}

The contrivance and development of techniques that can be used to investigate
the physics of very small system is currently attracting great interest \cite{Rit08}.
Examples of very small system are bio-molecular machines consisting of few or,
in some cases, even one molecule. These systems are able to ``efficiently'', in some sense, operate 
in non-equilibrium environments strongly affected by thermal fluctuations 
\cite{Haynie}.
For example, protein biosynthesis relies on the quality and efficiency with which the ribosome, 
a complex molecular machine, is able to pair mRNA codons with matching tRNA. 
During such process, known as ``decoding'', the ribosome and tRNA undergo large conformation changes 
which appear to correspond to an optimal in energy landscape recognition strategy \cite{SaTl13}.

One relevant motivation for attaining a precise characterizations of thermodynamic
efficiency of single-biomolecule systems such as molecular switches is the hope that their 
properties can be exploited in molecular-scale information processing \cite{Ben82}.
A long-standing conjecture by Landauer \cite{Lan61} surmises that the erasure of information 
is a dissipative process. Brownian computers \cite{Ben82} incorporate this conjecture in the form of 
models amenable to experimental \cite{BeArPeCiDiLu12} and theoretical \cite{AuGaMeMoMG12}
investigation.  In particular, the erasure of one bit of information can be modeled by steering
the evolution of the probability density of a diffusion process in a bistable potential by
manipulating the height and the depth of the potential wells \cite{Ben82}. Using the optimal control 
techniques introduced in the context of stochastic thermodynamics by \cite{AuMeMG11,AuMeMG12}, the 
authors of \cite{AuGaMeMoMG12} (see also \cite{Gaw13}) showed that the minimal average heat release 
during the erasure of one bit of information generically exceeds Landauer's $k_{\mathrm{B}}\,\beta^{-1}\ln 2$ 
bound where $\beta$ is the inverse temperature and $k_{\mathrm{B}}$ Boltzmann's constant. 
This theoretical result is consistent with the experimental findings of \cite{BeArPeCiDiLu12} where 
one bit erasure was modeled by a system of a single colloidal particle trapped in a modulated 
double-well potential.

The generality of Landauer's argument upholding an ultimate physical limit of irreversible computation 
and the existence of the recent experimental supporting evidences indicate 
that the conclusions of \cite{AuGaMeMoMG12} should remain true for Markov models more general than 
Langevin--Smoluchowski's dynamics. This intuition is corroborated by the fact that the entropy produced by
smoothly evolving the probability distribution of Markov jump processes between two assigned states is also subject 
to an analogous lower bound \cite{MGMePe12}.
The scope of the present contribution is to explore minimal entropy production transitions of continuous Markov 
processes governed by kinematic laws comprising a dissipative and a symplectic, conservative, component.
The consideration of such systems, commonly referred to as Langevin--Kramers or under-damped, is important
as they provide models of Newtonian mechanics in a heat bath. 

The structure of the paper is as follows. In section~\ref{sec:def} we describe the kinematic properties 
of a Langevin--Kramers diffusion. We also define the class $\mathbb{A}$ of \emph{admissible Hamiltonians} governing the
Langevin--Kramers dynamics to which we restrict our attention while considering the optimal control problem for 
the entropy production.
From the mathematical slant, it is obvious that an optimal control problem is well-posed if we assign besides
the ``cost functional'' to be minimized, the functional space of admissible controls. From the physics point 
of view, our aim is to explicitly restrict the attention to ``macroscopic'' control protocols modeled by smooth 
Hamiltonians acting on ``slower'' time scales as opposed to configurational degrees of freedom  subject to Brownian 
forces and fluctuating at the fastest time-scales in the model \cite{AlRiRi11} (see also discussion in \cite{MGMePe12}).  

In section~\ref{sec:tf} we briefly recall the stochastic thermodynamics of Langevin--Kramers diffusions drawing on
\cite{Se98,Sekimoto}. The main result is the expression of the entropy production $\mathcal{E}_{\hoz}$ over a
finite control horizon $[\ti,\tf]$ in terms of the current velocity of the Langevin--Kramers diffusion process.
As in the Langevin--Smoluchowski case \cite{JiangQianQian,AuMeMG12,PMG13}, the current velocity parametrization
plays a substantive role in unveiling the properties of the entropy production from both the thermodynamic 
and the control point of view. In the Langevin--Kramers case, the entropy production turns out to be a non-coercive
\cite{FlemingSoner} cost functional. Namely, the decomposition of the current velocity into dissipative and symplectic components
evinces that the entropy production is in fact a quadratic functional of the dissipative component alone. The origin of the
phenomenon is better understood by revisiting the probabilistic interpretation of the entropy production
which we do in section~\ref{sec:KL}. There we recall that $\mathcal{E}_{\hoz}$ coincides with the Kullback--Leibler 
divergence $\mathrm{K}(\mathrm{P}_{\chi} ||\mathrm{P}_{\tilde{\chi}})$ \cite{KuLe51} between the probability measure 
$\mathrm{P}_{\chi}$ of the primary Langevin--Kramers process $\chi$ and that $\mathrm{P}_{\tilde{\chi}}$ of a 
process $\tilde{\chi}$ obtained from the former by inverting the sign of the dissipative component of 
the drift and evolving in the opposite time direction. 
The Kullback--Leibler divergence is a relative entropy measuring the information loss occasioned when 
$\mathrm{P}_{\tilde{\chi}}$ is used to approximate $\mathrm{P}_{\chi}$. 
These facts substantiate, on the one hand, the identification of the entropy production as a natural indicator of the 
irreversibility of a physical process and, as such, as the embodiment of the second law of thermodynamics. On the other hand,
they pinpoint that the interpretation of the entropy production as a Kullback--Leibler divergence is possible in the
Kramers--Langevin case only by applying in the construction of the auxiliary process $\tilde{\chi}$ a different 
time reversal operation than in the Langevin--Smoluchowski case. Non-coercivity is the result of such time reversal operation.

The consequences of non-coercivity of the cost functional on the entropy production optimal control are, however, tempered 
by the regularity requirements imposed on the class $\mathbb{A}$ of admissible Hamiltonians. Simple considerations (section~\ref{sec:bound})
based on smoothness of the evolution show that the Langevin--Kramers entropy production must be bounded from below by the
entropy production generated by an optimally controlled Langevin--Smoluchowski diffusion connecting in the same horizon 
the configuration space marginals of the initial and final phase space probability densities. 
This result motivates the analysis of section~\ref{sec:Pontryagin} where we avail us of Pontryagin's maximum principle 
to directly investigate extremals of the entropy production in a finite 
time transition between assigned states. Pontryagin's maximum principle is formulated in terms of
Lagrange multipliers acting as conjugate ``momentum'' variables (see e.g. \cite{FlemingRishel,MalliarisBrock,Evans_OCT}).
We can therefore construe it as an ``Hamiltonian formulation'' of Bellman's optimal control theory which is based 
upon dynamic programming equations (see e.g. \cite{FlemingSoner,GuMo83,vHa07}). 
Relying on Pontryagin's maximum principle we conveniently arrive at the first main finding of the
present paper encapsulated in the extremal equations (\ref{Occam:eqs}). 
\emph{On the space $\mathbb{A}$ of admissible control Hamiltonians the entropy production generically attains an highly degenerate minimum value}. 
A distinctive feature of the extremal equations (\ref{Occam:eqs}) is that the coupling between the dynamic programming 
equation and the Fokker--Planck equation takes the \emph{non-local form} of a third auxiliary equation. We attribute
the occurrence of a non-local coupling to the divergenceless component of the Langevin--Kramers drift.

We illustrate these results in section~\ref{sec:ex} where we consider the explicitly solvable case of Gaussian statistics. 
While considering this example, we also inquire the recovery in the over-damped limit of the expression 
of the minimal entropy production by Langevin--Smoluchowski diffusions, a problem to which we systematically
turn in section~\ref{sec:multi}. There we derive our second main result: upon applying a multi-scale (homogenization)
asymptotic analysis \cite{BeLiPa78,PaSt08} we show that \emph{the cell problem associated to the extremal equations} (\ref{Occam:eqs})   
\emph{takes the form of a Monge--Amp\`ere--Kantorovich mass transport problem} \cite{Villani} for configuration
space marginals of the phase space probability densities.
The noteworthy aspect of this result is that the degeneracy of the extremals (\ref{Occam:eqs}) does not appear 
in the cell problem so that consistence with the results obtained for Langevin--Smoluchowski diffusions \cite{AuGaMeMoMG12} 
is guaranteed.

Finally, the last section \ref{sec:final} is devoted to a discussion and some conjectures concerning the existence
\emph{singular control strategies} which we explicitly ruled out while deriving the extremal equations (\ref{Occam:eqs}).

\section{From kinematics to dynamics}
\label{sec:def}

We consider a phase-space dynamics governed by 
\begin{subequations}
\label{def:process}
\begin{eqnarray}
\label{def:sde}
\mathrm{d}\boldsymbol{\chi}_{t}=
\left(\mathsf{J}-\mathsf{G}\right)\cdot\partial_{\boldsymbol{\chi}_{t}}H\frac{\dt}{\tau}
+\sqrt{\frac{2}{\beta\,\tau}} \,\mathsf{G}^{1/2}\cdot \mathrm{d}\boldsymbol{\omega}_{t}
\end{eqnarray}
\begin{eqnarray}
\label{def:init}
\mathrm{P}(\boldsymbol{x}\leq\boldsymbol{\chi}_{\ti}<\boldsymbol{x}+\mathrm{d}\boldsymbol{x} )
=\mathtt{m}_{\mathrm{o}}(\boldsymbol{x})\mathrm{d}^{2d}x
\end{eqnarray}
\end{subequations}
In (\ref{def:sde}) $\omega=\left\{\boldsymbol{\omega}_{t}\,,t\geq \ti\right\}$ denotes an $\mathbb{R}^{2\,d}$-valued Wiener-process while 
 $\mathsf{J}$ and $\mathsf{G}$ are $2\,d$-dimensional contravariant tensors of rank 2 with constant 
entries. In particular, $\mathsf{J}$ is the ``co-symplectic'' form
\begin{eqnarray}
\label{def:symplectic}
\mathsf{J}=\begin{bmatrix}
\mathsf{0} & \id_{d}
\\
-\id_{d} & \mathsf{0}
\end{bmatrix}
\hspace{1.0cm}\Rightarrow\hspace{1.0cm}
\mathsf{J}^{\dagger}\mathsf{J}=\id_{2\,d}
\end{eqnarray}
where $\id_{d}$ stands for the identity in $d$-dimensions. Furthermore, we associate 
to thermal noise fluctuations the constant pseudo-metric tensor
\begin{eqnarray}
\label{def:pm}
\mathsf{G}=\mathsf{P}_{v}
\end{eqnarray}
with $\mathsf{P}_{v}$ the vertical projector in phase space
\begin{eqnarray}
\label{}
\mathsf{P}_{h}\equiv\id_{d}\oplus \mathsf{0}
\hspace{1.0cm}\&\hspace{1.0cm}
\mathsf{P}_{v}\equiv\,\mathsf{0}\oplus\id_{d}
\end{eqnarray}
so that $\id_{2d}=\mathsf{P}_{h}\oplus\mathsf{P}_{v}$ and, for any $\boldsymbol{x}
=\boldsymbol{q}\oplus\boldsymbol{p}\in\mathbb{R}^{2\,d}$ with $\boldsymbol{p},\boldsymbol{q}\,\in\,\mathbb{R}^{d}$,
\begin{eqnarray}
\label{}
\begin{bmatrix}
\boldsymbol{q} \\ \boldsymbol{0}
\end{bmatrix}
\equiv\mathsf{P}_{h}\cdot\boldsymbol{x}
\hspace{1.0cm}\&\hspace{1.0cm}
\begin{bmatrix}
\boldsymbol{0} \\ \boldsymbol{p}
\end{bmatrix}\equiv\mathsf{P}_{v}\cdot\boldsymbol{x}
\end{eqnarray}
Finally, $\beta$ and $\tau$ are positive definite constants. We attribute to 
$\beta$ the physical interpretation of the inverse of the temperature 
and to $\tau$ that of the characteristic time scale of the system.
We will measure any other quantity encountered throughout the paper
in units of $\beta$ and $\tau$.

The kinematics in (\ref{def:sde}) satisfies the conditions required by 
H\"ormander theorem to prove that for any sufficiently regular, bounded 
from below and, growing sufficiently fast at infinity \emph{Hamilton function} $H$, 
the process $\chi=\left\{\boldsymbol{\chi}_{t}\,,t\in [\ti\,,\tf]\right\}$ admits a smooth transition probability density 
notwithstanding the degenerate form of the noise (see e.g. \cite{Rey06}). 
Furthermore, if $H$ is time independent, it is straightforward to verify that 
the measure relaxes to a steady state such that
\begin{subequations}
\label{def:measure}
\begin{eqnarray}
\label{def:spdf}
\mathrm{P}(\boldsymbol{x}\leq \boldsymbol{\chi}_{\infty}
<\boldsymbol{x}+\mathrm{d}\boldsymbol{x})=\beta^{d}\,
e^{\beta\,[F-H(\boldsymbol{x})]}\mathrm{d}^{2\,d}\boldsymbol{x}
\end{eqnarray}
\begin{eqnarray}
\label{def:fe}
F\equiv-\frac{1}{\beta}\ln\int_{\mathbb{R}^{2d}}
\mathrm{d}^{2\,d}x\,\beta^{d}\,e^{-\beta\,H(\boldsymbol{x})}
\end{eqnarray}
\end{subequations}
The expression of the normalization constant in (\ref{def:spdf}) befits
the interpretation of $F$ as \emph{equilibrium free energy}. 
For any finite time, we find expedient to write the probability
density of the system in the form
\begin{eqnarray}
\label{def:Mentropy}
\mathrm{P}(\boldsymbol{x}\leq \boldsymbol{\chi}_{t}
<\boldsymbol{x}+\mathrm{d}\boldsymbol{x})=\mathtt{m}(\boldsymbol{x},t)\,\mathrm{d}^{2\,d}\boldsymbol{x}
\equiv \beta^{d}\,
e^{-S(\boldsymbol{x},t)}\mathrm{d}^{2\,d}\boldsymbol{x}
\end{eqnarray}
We will refer to the non-dimensional function 
$S$ as the \emph{microscopic entropy} of the system inasmuch it specifies the amount of
information required to describe the state of the system given that the state
occurs with probability (\ref{def:Mentropy}) \cite{Sha48}.
The average variation of $S$ with respect to the measure of $\chi$, 
which we denote by $\mathrm{E}^{(\chi)}$,
\begin{eqnarray}
\label{def:Gibbs}
(\mathcal{S}_{t_{2}}-\mathcal{S}_{t_{1}})
=\mathrm{E}^{(\chi)}[S(\boldsymbol{\chi}_{t_{2}},t_{2})
-S(\boldsymbol{\chi}_{t_{1}},t_{1})]\hspace{1.0cm}
\forall\,t_{2}\geq t_{1}\geq\,0
\end{eqnarray}
specifies the variation of the \emph{Shannon-Gibbs entropy} of the system.
The representation (\ref{def:Mentropy}) of the probability density establishes 
an elementary link between the kinematics and the thermodynamics of the system. 
In order to describe the dynamics, we need to specify the Hamiltonian $H$. 
Our aim here is to determine $H$ by solving an optimal control problem associated 
to the minimization of a certain thermodynamic functional, the entropy production,
during a transition evolving the initial state (\ref{def:init}) into 
a final state
\begin{eqnarray}
\label{def:fin}
\mathrm{P}(\boldsymbol{x}\leq\boldsymbol{\chi}_{0}<\boldsymbol{x}+\mathrm{d}\boldsymbol{x} )
=\mathtt{m}_{\mathrm{f}}(\boldsymbol{x})\,\mathrm{d}^{2d}x
\end{eqnarray}
in a finite time horizon $[\ti\,,\tf]$. From this slant, we need 
to regard $H$ as the element of a class $\mathbb{A}$  of \emph{admissible controls}
comprising \emph{time}-dependent phase space functions
\begin{eqnarray}
\label{}
H\colon\mathbb{R}^{2\,d}\times\mathbb{R}_{+}\mapsto\mathbb{R}
\end{eqnarray}
satisfying the following requirements. Any $H\in\mathbb{A}$ must be at least
twice differentiable with respect to its phase space dependence and once 
differentiable for any $t\in [\ti\,,\tf]$. Furthermore, we require that for any Hamiltonian in $\mathbb{A}$ 
the evolution of the probability density of $\chi$ obeys a Fokker--Planck equation 
throughout the control horizon $[\ti\,,\tf]$. As we adopt the working hypothesis that
the initial and final state are described by probability densities integrable over
$\mathbb{R}^{2d}$, admissible Hamiltonians must then preserve this property for any $t\in\left[\ti\,,\tf\right]$.
These hypotheses entail
\begin{eqnarray}
\label{def:admissible}
\mathbb{A} \subset \mathbb{C}^{(2,1)}(\mathbb{R}^{2d},[\ti,\tf])\cap\mathbb{L}^{(2)}(\mathbb{R}^{2d},\mathtt{m}\,\mathrm{d}^{2d}x)
\end{eqnarray}
where $\mathtt{m}\,\mathrm{d}^{2d}x$  portends square integrability requirement with respect
to the density of (\ref{sec:def}). We will reserve the simpler notation $\mathbb{L}^{(2)}(\mathbb{R}^{2d})$
to the space of functions square integrable with respect to the Lebesgue measure.

\subsection{Generator of the process and ``metriplectic'' structure}

The scalar generator  $\mathfrak{L}$ of (\ref{def:process}) 
acts on any differentiable phase space function $f$ as
\begin{eqnarray}
\label{def:gen}
(\mathfrak{L}f)(\boldsymbol{x},t)\equiv\left\{
\left[\left(\mathsf{J}-\mathsf{G}\right)\cdot\partial_{\boldsymbol{x}}H(\boldsymbol{x},t)\right]\cdot\partial_{\boldsymbol{x}}
+\frac{1}{\beta} \,\mathsf{G}:\partial_{\boldsymbol{x}}\otimes\partial_{\boldsymbol{x}}\right\}f(\boldsymbol{x},t)
\end{eqnarray}
In (\ref{def:gen}) and in what follows we use the notation
\begin{eqnarray}
\label{def:matrix}
\mathsf{A}:\mathsf{B}\equiv\Tr\mathsf{A}^{\dagger}\mathsf{B}
\end{eqnarray}
for the scalar product between matrices. 
It is worth noticing that it is possible to write the generator in terms 
of a symplectic and a ``metric'' bracket operation. Namely, the transpose 
$\mathsf{J}^{\dagger}$ of $\mathsf{J}$ defines a \emph{symplectic form} between
differentiable phase space functions  $f_{i}$, $i=1,2$ 
\begin{eqnarray}
\label{def:Poisson}
\brl f_{1}\,,f_{2}\brr_{\mathsf{J}^{\dagger}}\equiv 
(\partial_{\boldsymbol{x}}f_{1})\cdot(\mathsf{J}^{\dagger}\cdot\partial_{\boldsymbol{x}}f_{2})
\equiv
\partial_{\boldsymbol{p}}f_{1}\overset{d}{\cdot}\partial_{\boldsymbol{q}}f_{2}
-\partial_{\boldsymbol{q}}f_{1}\overset{d}{\cdot}\partial_{\boldsymbol{p}}f_{2}
\end{eqnarray}
coinciding with \emph{Poisson brackets} in a Darboux chart.
In (\ref{def:gen}) and (\ref{def:Poisson}) the symbol ``$\cdot$'' 
stands for the dot-product in $\mathbb{R}^{2d}$ and $\overset{d}{\cdot}$ for the analogous operation in $\mathbb{R}^{d}$.
Similarly, it is possible to associate to $\mathsf{G}$ the degenerate metric brackets
\begin{eqnarray}
\label{def:pseudomb}
\brl f_{1}\,,f_{2}\brr_{\mathsf{G}}\equiv(\partial_{\boldsymbol{x}}f_{1})
\cdot(\mathsf{G}\cdot\partial_{\boldsymbol{x}}f_{2})
\end{eqnarray}
acting on scalars in $\mathbb{R}^{2d}$ and the pseudo-norm
\begin{eqnarray}
\label{}
\parallel\boldsymbol{f}\parallel_{\mathsf{G}}^{2}\equiv\boldsymbol{f}\cdot\mathsf{G}\cdot\boldsymbol{f}
\hspace{1.5cm}\forall\,\boldsymbol{f}\in\mathbb{R}^{2d}
\end{eqnarray}
The generator becomes
\begin{eqnarray}
\label{def:genp}
\mathfrak{L}f=\brl H\,,f \brr_{\mathsf{J}^{\dagger}}+\mathfrak{G}f
\end{eqnarray}
where
\begin{eqnarray}
\label{def:gendiss}
\mathfrak{G}f
=-\brl H\,,f\brr_{\mathsf{G}}+\frac{1}{\beta}\mathsf{G}:\partial_{\boldsymbol{x}}\otimes\partial_{\boldsymbol{x}}f
\end{eqnarray}
The Poisson brackets embody the energy conserving component
of the kinematics. The differential operation $\mathfrak{G}$ describes dissipation 
which occurs via a deterministic friction mechanism associated
to the metric brackets and via thermal interactions encapsulated
in the second order differential terms. In the analytic mechanics literature
it is customary to refer to systems whose generator comprises a symplectic 
and a metric structure as ``\emph{metriplectic}'' see e.g. \cite{Mor86,Mor09}.

The $\mathbb{L}^{(2)}(\mathbb{R}^{2d})$ adjoint of $\mathfrak{L}$ with respect to the Lebesgue 
measure governs the evolution of the probability density $\mathtt{m}$ of the system. 
The anti-symmetry of the Poisson brackets yields
\begin{eqnarray}
\label{def:FP}
\mathfrak{L}^{\dagger}f=-\brl H\,,f\brr_{\mathsf{J}^{\dagger}}+\mathfrak{G}^{\dagger}f
\end{eqnarray}
with $\mathfrak{G}^{\dagger}$ the $\mathbb{L}^{(2)}(\mathbb{R}^{2d})$-adjoint of 
(\ref{def:gendiss}).

\section{Thermodynamic functionals}
\label{sec:tf}

Following \cite{Se98}, we identify the \emph{heat} released during individual realizations
of $\chi$ with the Stratonovich stochastic integral
\begin{eqnarray}
\label{tf:heat}
Q_{\hoz}=-\int_{\ti}^{\tf}\mathrm{d}\boldsymbol{\chi}_{t}\str{\cdot}\partial_{\boldsymbol{\chi}_{t}}H
\end{eqnarray}
The $1/2$ betokens Stratonovich's mid-point convention.
If in conjunction with (\ref{tf:heat}) we define the
\emph{work} as
\begin{eqnarray}
\label{check:work}
W_{\hoz}=\int_{\ti}^{\tf}\dt\,\partial_{t}H
\end{eqnarray}
we recover the first law of thermodynamics in the form 
\begin{eqnarray}
\label{}
(W-Q)_{\hoz}=\int_{\ti}^{\tf}\left(\dt\,\partial_{t}H+d\boldsymbol{\chi}_{t}\str{\cdot}\partial_{\boldsymbol{\chi}_{t}}H\right)
=H_{\tf}-H_{\ti}
\end{eqnarray}
As our working hypotheses allow us to perform integrations by parts 
without generating boundary terms, the definition of the Stratonovich integral begets 
(see e.g. \cite{Nelson85} pag. 33) the equality 
\begin{eqnarray}
\label{tf:derH}
\mathrm{E}^{(\chi)}\int_{\ti}^{\tf}\mathrm{d}\boldsymbol{\chi}_{t}\str{\cdot}\partial_{\boldsymbol{\chi}_{t}}H=
\mathrm{E}^{(\chi)}\int_{\ti}^{\tf}\frac{\dt}{\tau} \,\boldsymbol{v}_{t}\cdot\partial_{\boldsymbol{\chi}_{t}}H
\end{eqnarray}
for 
\begin{eqnarray}
\label{tf:cv}
\boldsymbol{v}
=\mathsf{J}\cdot\partial_{\boldsymbol{x}}H
-\mathsf{G}\cdot\partial_{\boldsymbol{x}}
\left(H-\frac{1}{\beta}S\right)
\end{eqnarray}
the \emph{current velocity} \cite{Nelson01} (see also \ref{ap:derivatives}) and $S$ the microscopic entropy 
(\ref{def:Mentropy}).
Upon inserting (\ref{tf:cv}) into (\ref{tf:derH}) after straightforward algebra we arrive at
\begin{eqnarray}
\label{tf:entropy}
\hspace{-1.0cm}
\frac{\mathcal{E}_{\hoz}}{\beta}\,\equiv\mathrm{E}^{(\chi)}\left\{Q_{\hoz}+\frac{S_{\tf}-S_{\ti}}{\beta}\right\}
=\mathrm{E}^{(\chi)}\int_{\ti}^{\tf}\frac{\dt}{\tau} \parallel\partial_{\boldsymbol{\chi}_{t}}A\parallel_{\mathsf{G}}^{2}\,\geq\,0
\end{eqnarray}
We interpret the phase space function
\begin{eqnarray}
\label{tf:Hel}
A=H-\frac{1}{\beta}S
\end{eqnarray}
as the ``\emph{non-equilibrium} Helmholtz energy density'' of the system and, 
the non-dimensional quantity $\mathcal{E}$ as the \emph{entropy production} during the transition 
(see e.g. \cite{MaReMo00,JiangQianQian,ChGa08}). The interpretation is upheld by 
observing that the entropy production rate, $\mathrm{E}^{(\chi)}\parallel\partial_{\boldsymbol{\chi}_{t}}A\parallel_{\mathsf{G}}^{2}$ , is a positive definite 
quantity generically vanishing only at equilibrium. 
On this basis, we regard the inequality (\ref{tf:entropy}) as the embodiment of the \emph{second law}
of thermodynamics. Furthermore, the positive definiteness of the entropy production yields a Jarzynski 
type \cite{Jar97} bound for the mean work
\begin{eqnarray}
\label{}
\mathrm{E}^{(\chi)}W_{\hoz}=\mathrm{E}^{(\chi)}\left(A_{\tf}-A_{\ti}\right)+\frac{1}{\beta}\mathcal{E}_{\hoz}
\geq\mathrm{E}^{(\chi)}\left(A_{\tf}-A_{\ti}\right)
\end{eqnarray}

\section{Probabilistic interpretation of the thermodynamic functionals }
\label{sec:KL}

The entropy production (\ref{tf:entropy}) admits an intrinsic information 
theoretic interpretation as a quantifier of the irreversibility of a transition. Namely
it coincides with the Kullback--Leibler divergence between
the measure of the process (\ref{def:process}) and that of the \emph{backward-in-time}
diffusion process $\tilde{\chi}=\left\{\boldsymbol{\tilde{\chi}}_{t};
t\in[\ti,\tf]\right\}$ obtained by reverting the sign of the dissipative component 
of the drift:
\begin{subequations}
\label{KL:process}
\begin{eqnarray}
\label{KL:sde}
\mathrm{d}\boldsymbol{\tilde{\chi}}_{t}=
\left(\mathsf{J}+\mathsf{G}\right)\cdot\partial_{\boldsymbol{\tilde{\chi}}_{t}}H\frac{\dt}{\tau}
+\sqrt{\frac{2}{\beta\,\tau}} \,\mathsf{G}^{1/2}\cdot \mathrm{d}\boldsymbol{\omega}_{t}
\end{eqnarray}
\begin{eqnarray}
\label{KL:init}
\mathrm{P}(\boldsymbol{x}\leq\boldsymbol{\tilde{\chi}}_{\tf}<\boldsymbol{x}+\mathrm{d}\boldsymbol{x} )
=\mathtt{m}(\boldsymbol{x},\tf)\mathrm{d}^{2d}x
\end{eqnarray}
\end{subequations}
In (\ref{KL:init}) $\mathtt{m}(\boldsymbol{x},\tf)$ 
is the probability density generated by (\ref{def:process}) evaluated
at $\tf$ whilst $H$ in (\ref{KL:sde}) is the very same Hamiltonian entering
(\ref{def:process}). The drift in (\ref{KL:sde})
\emph{must} be interpreted as the mean backward derivative $\mathrm{D}_{\boldsymbol{\tilde{\chi}}_{t}}^{-}\boldsymbol{\tilde{\chi}}_{t}$ of the process
$\tilde{\chi}$ (see \ref{ap:derivatives} for details).
In order to compare $\chi$ with $\tilde{\chi}$ we suppose that the corresponding
probability measures $\mathrm{P}_{\chi}$ and $\mathrm{P}_{\tilde{\chi}}$ have support
over the same Borel sigma algebra $\mathcal{F}_{[\ti,\tf]}$ and are absolutely continuous
with respect to the Lebesgue measure. The difference between $\chi$ and
$\tilde{\chi}$ consists then in the fact that for any $t\in[\ti,\tf]$, $\boldsymbol{\chi}_{t}$
is adapted (i.e. measurable with respect) to the sub-sigma algebra $\mathcal{F}_{[\ti,t]}$
of  $\mathcal{F}_{[\ti,\tf]}$ comprising all ``past'' events at time $t$. 
The realization $\boldsymbol{\tilde{\chi}}_{t}$ of $\tilde{\chi}$ is instead
adapted  to the sub-sigma algebra $\mathcal{F}_{[t,\tf]}$
of  $\mathcal{F}_{[\ti,\tf]}$ comprising all ``future'' events at time $t$ (see e.g. \cite{Za86} for details).
A tangible consequence of this difference is that for any integrable test 
vector field $\boldsymbol{V}\colon\mathbb{R}^{2d}\mapsto\mathbb{R}^{2d}$ and any $t\in[\ti,\tf]$
the Ito \emph{pre-point} stochastic integral satisfies
\begin{eqnarray}
\label{KL:pre}
\mathrm{E}^{(\chi)}\int_{\ti}^{t}\mathrm{d}\boldsymbol{\chi}_{t}\ito{\cdot}\boldsymbol{V}(\boldsymbol{\chi}_{t})=
\mathrm{E}^{(\chi)}\int_{\ti}^{t}\mathrm{d}t\, [\left(\mathsf{J}-\mathsf{G}\right)
\cdot\partial_{\boldsymbol{\chi}_{t}}H](\boldsymbol{\chi}_{t})
\cdot\boldsymbol{V}(\boldsymbol{\chi}_{t})
\end{eqnarray}
while instead the \emph{post-point} prescription yields
\begin{eqnarray}
\label{KL:post}
\mathrm{E}^{(\tilde{\chi})}\int_{\ti}^{t}\mathrm{d}\boldsymbol{\tilde{\chi}}_{t}
\oti{\cdot}\boldsymbol{V}(\boldsymbol{\tilde{\chi}}_{t})=
\mathrm{E}^{(\tilde{\chi})}\int_{\ti}^{t}\mathrm{d}t\, [\left(\mathsf{J}+\mathsf{G}\right)
\cdot\partial_{\boldsymbol{\tilde{\chi}}_{t}}H](\boldsymbol{\tilde{\chi}}_{t})
\cdot\boldsymbol{V}(\boldsymbol{\tilde{\chi}}_{t})
\end{eqnarray}
We will now avail us of these observations to prove that
\begin{proposition}
\label{prop:KL}
if we decompose the current velocity (\ref{tf:cv}) into a 
``dissipative'' component
\begin{eqnarray}
\label{KL:cvd}
\boldsymbol{v}_{+}(\boldsymbol{x},t)=-\mathsf{G}\cdot\partial_{\boldsymbol{x}}A(\boldsymbol{x},t)
\end{eqnarray}
and a divergenceless ``conservative'' component
\begin{eqnarray}
\label{KL:cvc}
\boldsymbol{v}_{-}(\boldsymbol{x},t)=\mathsf{J}\cdot\partial_{\boldsymbol{x}}H(\boldsymbol{x},t)
\end{eqnarray}
then Kullback--Leibler divergence between $\mathrm{P}_{\tilde{\chi}}$ and  $\mathrm{P}_{\chi}$ depends only
on $\boldsymbol{v}_{+}$ and is equal to
\begin{eqnarray}
\label{KL:main}
\mathrm{K}(\mathrm{P}_{\chi}||\mathrm{P}_{\tilde{\chi}})\equiv 
\mathrm{E}^{(\chi)}\ln\frac{\mathrm{d}\mathrm{P}_{\chi}}{\mathrm{d}\mathrm{P}_{\tilde{\chi}}}=\beta\,
\mathrm{E}^{(\chi)}\int_{\ti}^{\tf}\frac{\mathrm{d}t}{\tau}\parallel\partial_{\boldsymbol{\chi}_{t}}A\parallel_{\mathsf{G}}^{2}
\end{eqnarray}
\end{proposition}
\begin{proof}~\newline\noindent
The proof proceeds in two steps: first we introduce two auxiliary diffusion processes
one forward and the other backward in time for which we know the expression of the Radon-Nikodym
derivative of the corresponding probability measures; then we apply Cameron--Martin--Girsanov's formula 
(see e.g. \cite{Klebaner}) to relate the auxiliary processes to $\chi$ and $\tilde{\chi}$.
\begin{enumerate}
\item[\textbf{First step}] We call $\eta=\left\{\boldsymbol{\eta}_{t}; t\in[\ti,\tf]\right\}$ the diffusion with
 $\mathcal{F}_{[\ti,t]}$-adapted realizations solution of the \emph{forward} 
stochastic dynamics
\begin{subequations}
\label{KL:faux}
\begin{eqnarray}
\label{KL:faux1}
\mathrm{d}\boldsymbol{\eta}_{t}=
\mathsf{J}\cdot\partial_{\boldsymbol{\eta}_{t}}H\frac{\dt}{\tau}
+\sqrt{\frac{2}{\beta\,\tau}} \,\mathsf{G}^{1/2}\cdot \mathrm{d}\boldsymbol{\omega}_{t}
\end{eqnarray}
\begin{eqnarray}
\label{KL:faux2}
\mathrm{P}(\boldsymbol{x}\leq\boldsymbol{\eta}_{\ti}<\boldsymbol{x}+\mathrm{d}\boldsymbol{x} )
=\mathtt{m}_{\mathrm{o}}(\boldsymbol{x})\mathrm{d}^{2d}x
\end{eqnarray}
\end{subequations}
Similarly, let  $\tilde{\eta}=\left\{\boldsymbol{\tilde{\eta}}_{t}; t\in[\ti,\tf]\right\}$  diffusion governed by the \emph{backward} dynamics
\begin{subequations}
\label{KL:baux}
\begin{eqnarray}
\label{KL:baux1}
\mathrm{d}\boldsymbol{\tilde{\eta}}_{t}=
\mathsf{J}\cdot\partial_{\boldsymbol{\tilde{\eta}}_{t}}H\frac{\dt}{\tau}
+\sqrt{\frac{2}{\beta\,\tau}} \,\mathsf{G}^{1/2}\cdot \mathrm{d}\boldsymbol{\omega}_{t}
\end{eqnarray}
\begin{eqnarray}
\label{KL:baux2}
\mathrm{P}(\boldsymbol{x}\leq\boldsymbol{\tilde{\chi}}_{\tf}<\boldsymbol{x}+\mathrm{d}\boldsymbol{x} )
=\mathtt{m}(\boldsymbol{x},\tf)\mathrm{d}^{2d}x
\end{eqnarray}
\end{subequations}
with $\boldsymbol{\tilde{\eta}}_{t}$ $\mathcal{F}_{[t,\tf]}$-adapted.
The simultaneous occurrence of additive noise and divergence-less drift in
(\ref{KL:faux1}),  (\ref{KL:baux1}) occasions the identity
\begin{eqnarray}
\label{}
\mathtt{p}_{\eta}(\boldsymbol{x}_{2},t_{2}|\boldsymbol{x}_{1},t_{1})=
\mathtt{p}_{\tilde{\eta}}(\boldsymbol{x}_{1},t_{1}|\boldsymbol{x}_{2},t_{2})
\end{eqnarray}
satisfied by the transition probability densities of $\eta$ and $\tilde{\eta}$ for all 
$\boldsymbol{x}_{1},\boldsymbol{x}_{2}\in\mathbb{R}^{2d}$ and for all $t_{1}, t_{2}\in [\ti,\tf]$, $t_{1}\,\leq\,t_{2}$.
We therefore conclude that 
\begin{eqnarray}
\label{}
\frac{\mathrm{d}\mathrm{P}_{\tilde{\eta}}}{\mathrm{d}\mathrm{P}_{\eta}}=
\frac{\mathtt{m}(\boldsymbol{\eta}_{\tf},\tf)}{\mathtt{m}_{o}(\boldsymbol{\eta}_{\ti})}
\end{eqnarray}
\item[\textbf{Second step}] We apply the composition 
property of the Radon--Nikodym derivative in order to couch (\ref{KL:main})
into the form
\begin{eqnarray}
\label{KL:KL1}
\mathrm{K}(\mathrm{P}_{\chi}||\mathrm{P}_{\tilde{\chi}})
=-
\mathrm{E}^{(\eta)}\frac{\mathrm{d}\mathrm{P}_{\chi}}
{\mathrm{d}\mathrm{P}_{\eta}}
\ln \left(
\frac{\mathrm{d}\mathrm{P}_{\tilde{\chi}}}{\mathrm{d}\mathrm{P}_{\eta}}\middle /
\frac{\mathrm{d}\mathrm{P}_{\chi}}{\mathrm{d}\mathrm{P}_{\eta}}
\right)
\end{eqnarray}
Cameron--Martin--Girsanov's formula yields immediately
\begin{eqnarray}
\label{KL:CMG1}
\hspace{-2.2cm}
\frac{\mathrm{d}\mathrm{P}_{\chi} }{\mathrm{d} \mathrm{P}_{\eta}}=
\exp\left\{\frac{\beta}{2}\int_{\ti}^{\tf}
\left[-\left(\mathsf{G}\cdot\partial_{\boldsymbol{\eta}_{t}}H\right)\ito{\cdot}\left(\mathrm{d}\boldsymbol{\eta}_{t}
-\frac{\mathrm{d}t}{\tau}\mathsf{J}\cdot\partial_{\boldsymbol{\eta}_{t}}H\right)
-\frac{\mathrm{d}t}{\tau}\,\parallel\partial_{\boldsymbol{\eta}_{t}}H\parallel_{\mathsf{G}}^{2}\right]\right\}
\end{eqnarray}
which by is a martingale at time $\tf$ by (\ref{KL:pre}). 
In order to compute $\mathrm{d}\mathrm{P}_{\tilde{\chi}}/\mathrm{d}P_{\eta}$ we first use Cameron--Martin--Girsanov's formula 
adapted to backward processes \cite{Mey82}
\begin{eqnarray}
\label{KL:CMG1}
\hspace{-2.2cm}
\frac{\mathrm{d}\mathrm{P}_{\chi} }{\mathrm{d} \mathrm{P}_{\eta}}=
\exp\left\{\frac{\beta}{2}\int_{\ti}^{\tf}
\left[\left(\mathsf{G}\cdot\partial_{\boldsymbol{\tilde{\eta}}_{t}}H\right)\oti{\cdot}
\left(\mathrm{d}\boldsymbol{\tilde{\eta}}_{t}
-\frac{\mathrm{d}t}{\tau}\mathsf{J}\cdot\partial_{\boldsymbol{\tilde{\eta}}_{t}}H\right)
-\frac{\mathrm{d}t}{\tau}\,\parallel\partial_{\boldsymbol{\tilde{\eta}}_{t}}H\parallel_{\mathsf{G}}^{2}\right]\right\}
\end{eqnarray}
Then we apply again the composition property to write
\begin{eqnarray}
\label{KL:CMG3}
\lefteqn{\hspace{-2.2cm}
\frac{\mathrm{d}\mathrm{P}_{\tilde{\chi}} }{\mathrm{d} \mathrm{P}_{\eta}}
=\frac{\mathrm{d}\mathrm{P}_{\tilde{\chi}} }{\mathrm{d} \mathrm{P}_{\tilde{\eta}}}
\frac{\mathrm{d}\mathrm{P}_{\tilde{\eta}} }{\mathrm{d} \mathrm{P}_{\eta}}
=
}
\nonumber\\&&
\hspace{-1.8cm}
\frac{\mathtt{m}(\boldsymbol{\eta}_{\tf},\tf)}{\mathtt{m}_{o}(\boldsymbol{\eta}_{\ti})}
\exp\left\{\frac{\beta}{2}\int_{\ti}^{\tf}
\left[\left(\mathsf{G}\cdot\partial_{\boldsymbol{\eta}_{t}}H\right)\oti{\cdot}
\left(\mathrm{d}\boldsymbol{\eta}_{t}
-\frac{\mathrm{d}t}{\tau}\mathsf{J}\cdot\partial_{\boldsymbol{\eta}_{t}}H\right)
-\frac{\mathrm{d}t}{\tau}\,\parallel\partial_{\boldsymbol{\eta}_{t}}H\parallel_{\mathsf{G}}^{2}\right]\right\}
\end{eqnarray}
since $\boldsymbol{\eta}_{t}$ on the right hand side plays the role of a mute integration variable. 
Upon inserting (\ref{KL:CMG1}) and (\ref{KL:CMG3}) in (\ref{KL:KL1}) and expressing
the stochastic integrals into the time-reversal invariant Stratonovich
mid-point discretization we arrive at
\begin{eqnarray}
\label{KL:KL2}
\lefteqn{
\ln\frac{\mathrm{d}\mathrm{P}_{\tilde{\chi}} }{\mathrm{d} \mathrm{P}_{\chi}}=
\int_{\ti}^{\tf}\frac{\mathrm{d}t}{\tau}\left\{\left[(\mathsf{J}\cdot\partial_{\boldsymbol{\chi}_{t}}H)\cdot
\partial_{\boldsymbol{\chi}_{t}}
+\tau\,\partial_{t}
\right]\ln\frac{\mathtt{m}}{\beta^{d}}\right\}
}
\nonumber\\&&
+\beta\int_{\ti}^{\tf}\left[
\mathrm{d}\boldsymbol{\chi}_{t}
-\frac{\mathrm{d}t}{\tau}(\mathsf{J}\cdot\partial_{\boldsymbol{\chi}_{t}}H)
\right]\str{\cdot}
\left(\mathsf{G}\cdot\partial_{\boldsymbol{\chi}_{t}}H
+\frac{1}{\beta}\partial_{\boldsymbol{\chi}_{t}} \ln\frac{\mathtt{m}}{\beta^{d}}
\right)
\end{eqnarray}
The first integral vanishes on average since
\begin{eqnarray}
\label{}
\hspace{-1.8cm}
\mathrm{E}^{(\chi)}\left\{\left[(\mathsf{J}\cdot\partial_{\boldsymbol{\chi}_{t}}H)\cdot
\partial_{\boldsymbol{\chi}_{t}}
+\tau\,\partial_{t}
\right]\ln\frac{\mathtt{m}}{\beta^{d}}\right\}=\int_{\mathbb{R}^{2d}}\mathrm{d}^{2d}x\,
\left(\partial_{\boldsymbol{x}}\cdot\boldsymbol{v}+\tau\,\partial_{t}\right)\mathtt{m}=0
\end{eqnarray}
In virtue of the properties of the Stratonovich integral (see e.g. \cite{Nelson85} pag. 33), 
the expectation value of the second integral in (\ref{KL:KL2}) yields
\begin{eqnarray}
\label{}
\lefteqn{\hspace{-2.2cm}
\mathrm{E}^{(\chi)}\int_{\ti}^{\tf}\left[
\mathrm{d}\boldsymbol{\chi}_{t}
-\frac{\mathrm{d}t}{\tau}(\mathsf{J}\cdot\partial_{\boldsymbol{\chi}_{t}}H)
\right]\str{\cdot}
\left(\mathsf{G}\cdot\partial_{\boldsymbol{\chi}_{t}}H
+\frac{1}{\beta}\partial_{\boldsymbol{\chi}_{t}} \ln\frac{\mathtt{m}}{\beta^{d}}
\right)
}
\nonumber\\&&
\hspace{-1.8cm}
=\mathrm{E}^{(\chi)}\int_{\ti}^{\tf}\frac{\mathrm{d}t}{\tau}\boldsymbol{v}_{+}\cdot
\left(\mathsf{G}\cdot\partial_{\boldsymbol{\chi}_{t}}H
+\frac{1}{\beta}\partial_{\boldsymbol{\chi}_{t}} \ln\frac{\mathtt{m}}{\beta^{d}}
\right)
=-\mathrm{E}^{(\chi)}\int_{\ti}^{\tf}\frac{\mathrm{d}t}{\tau}
\parallel\partial_{\boldsymbol{\chi}_{t}}A\parallel_{\mathsf{G}}^{2}
\end{eqnarray}
where the last equality holds because $\mathsf{G}$ is a projector.
\end{enumerate}
\end{proof}
Some remarks are in order. 
\begin{enumerate}
\item The information theoretic interpretation of the entropy production is a consequence 
of the fluctuation relation type \cite{EvSe94,GaCo95,Jar97,Kur98,LeSp99,Cro99,MaReMo00,ChGa08} 
equality (\ref{KL:KL2}). Reference \cite{ChGu10} discusses in details the relation between fluctuation 
relations for Markov processes and exponential martingales. Finally, a recent nice overview of
fluctuation theorems can be found in the lectures \cite{Gaw13}.
\item The proof of the identity (\ref{KL:main}) is based on
the comparison between a forward and a backward dynamics in the sense of Nelson \cite{Nelson85,Nelson01}
and admits a straightforward generalization to all the cases discussed in \cite{ChGa08}.
In particular, choosing the auxiliary process $\eta$ to be the stochastic development map 
(see e.g. \cite{Cou11}) yields readily covariant expressions for the entropy production 
by diffusion on Riemann manifolds \cite{JiangQianQian,PMG13}.
\item The stochastic development map in the Euclidean case with flat metric reduces 
to the standard Wiener process. An alternative proof of (\ref{KL:main}) can be then 
obtained by taking the limit of vanishing noise acting on the position coordinate 
process.
\item The dissipative (\ref{KL:cvd}) and conservative (\ref{KL:cvc}) components of the
current velocity are \emph{not} $\mathbb{L}^{(2)}(\mathbb{R}^{2d},\mathtt{m}\mathrm{d}^{2d}x)$-orthogonal. Therefore, it is not natural to 
regard the dissipative component as an independent control of the entropy production.
\end{enumerate}

\section{A general bound for the entropy production from moments equation}
\label{sec:bound}

The main consequence of the last remark the foregoing section is that the Hamiltonian $H$ is the 
natural control functional for the entropy production. The entropy production
is, however, independent of derivatives of $H$ with respect to position coordinates. This fact 
poses the question whether the uncoerced degrees of freedom can be used to steer a smooth Langevin--Kramers
dynamics to accomplish a finite-time transition between assigned states for arbitrarily low values 
of the entropy production. 
A simple lower bound provided by the ``macroscopic'', in kinetic theory sense 
(see e.g. \cite{Struchtrup}), dynamics shows that this cannot be the case.
Let
\begin{eqnarray}
\label{bound:marginal}
\tilde{\mathtt{m}}(\boldsymbol{q},t)\equiv\int_{\mathbb{R}^{d}}\mathrm{d}^{d}p\,\mathtt{m}(\boldsymbol{p},\boldsymbol{q},t) 
\end{eqnarray}
the marginal probability density over the configuration space of (\ref{def:process}).
Integrating the Fokker-Planck equation governing the evolution of $\mathtt{m}$ over 
momenta it is readily seen that $\tilde{\mathtt{m}}$ obeys
\begin{eqnarray}
\label{bound:transport}
\tau\,\partial_{t}\tilde{\mathtt{m}}+\partial_{\boldsymbol{q}}\cdot \tilde{\mathtt{m}}\,\tilde{\boldsymbol{v}}=0
\end{eqnarray}
We define the ``macroscopic drift''  $\tilde{\boldsymbol{v}}$ as the average
\begin{eqnarray}
\label{ap:macrov}
(\tilde{\mathtt{m}}\,\tilde{\boldsymbol{v}})(\boldsymbol{q},t)\equiv\int_{\mathbb{R}^{d}}\mathrm{d}^{d}p\,
(\mathtt{m}\,\partial_{\boldsymbol{p}}A)(\boldsymbol{p},\boldsymbol{q},t)
\end{eqnarray}
over the momentum gradient of the non-equilibrium Helmholtz energy density (\ref{tf:Hel}). 
Let $\boldsymbol{\tilde{V}}\equiv 0\oplus \boldsymbol{\tilde{v}}$ the phase space lift of $\boldsymbol{\tilde{v}}$. 
Since $\mathsf{G}$ is the vertical projector in $\mathbb{R}^{2d}$, an immediate consequence of (\ref{ap:macrov}) is the inequality
\begin{eqnarray}
\label{bound:inequality}
\mathcal{E}_{\hoz}&=&\int_{\ti}^{\tf}\mathrm{d}t\int_{\mathbb{R}^{2d}}\mathrm{d}^{2d}x\,\mathtt{m}\,
\left(\parallel\partial_{\boldsymbol{x}}A-\tilde{\boldsymbol{V}}\parallel_{\mathsf{G}}^{2}
+\parallel\tilde{\boldsymbol{V}}\parallel_{\mathsf{G}}^{2}\right)
\nonumber\\
&\geq& \int_{\ti}^{\tf}\mathrm{d}t\int_{\mathbb{R}^{d}}\mathrm{d}^{d}q\,
\tilde{\mathtt{m}}\,\parallel\tilde{\boldsymbol{v}}\parallel_{\id_{d}}^{2}
=\tilde{\mathcal{E}}_{\hoz}
\end{eqnarray}
for $\parallel\parallel_{\id_{d}}$ the Euclidean norm in $\mathbb{R}^{d}$.
Taking into account that (\ref{bound:transport}) must also hold true, we interpret $\boldsymbol{\tilde{v}}$
as the current velocity of an effective Langevin--Smoluchowski dynamics. Furthermore,
$\tilde{\mathcal{E}}_{\hoz}$ attains a minimum if the pair $(\tilde{\mathtt{m}}\,,\tilde{\boldsymbol{v}})$ 
is determined from the solution of Monge--Amp\`ere--Kantorovich 
problem \cite{AuMeMG12,AuGaMeMoMG12}. We will see in section~\ref{sec:multi} that the 
bound becomes tight in the presence of a strong separation of scales between position 
and momentum dynamics.

\section{Entropy production extremals via Pontryagin theory}  
\label{sec:Pontryagin}

The existence of the general bound (\ref{bound:inequality}) indicates that
the question of existence of entropy production extremals in the admissible class $\mathbb{A}$ 
(\ref{def:admissible}) is well posed. In order to directly pursue the quest, we introduce 
the Pontryagin functional \cite{FlemingRishel} 
\begin{eqnarray}
\label{Pontryagin:action}
\hspace{-1.0cm}
\mathcal{A}(\mathtt{m},V,A)=\int_{\ti}^{\tf}\frac{\mathrm{d}t}{\tau}\,\int_{\mathbb{R}^{2\,d}}
\mathrm{d}^{2\,d}x\,
\left\{
\mathtt{m}\,\parallel\partial_{\boldsymbol{x}}A\parallel_{\mathsf{G}}^{2}
-V\left(\tau\,\partial_{t}-\mathfrak{L}^{\dagger}\right)\mathtt{m}
\right\}
\end{eqnarray}
complemented by the boundary conditions
\begin{eqnarray}
\label{Pontryagin:bc}
\mathtt{m}(\boldsymbol{x},\ti)=\mathtt{m}_{\mathrm{o}}(\boldsymbol{x})
\hspace{1.0cm}\&\hspace{1.0cm}
\mathtt{m}(\boldsymbol{x},\tf)=\mathtt{m}_{\mathrm{f}}(\boldsymbol{x})
\end{eqnarray}
The functional (\ref{Pontryagin:action}) specifies a generalized entropy production in which 
the dynamical constraint on the probability density appears explicitly. 
The ``costate'' field $V\colon\mathbb{R}^{2d}\times[\ti,\tf]\mapsto\mathbb{R}$ is a Lagrange multiplier imposing the probability 
density $\mathtt{m}$ to evolve according the Fokker--Planck of (\ref{def:process}). The sign convention 
of $V$ suits the identification of the extremal value of the costate with the ``value'' or ``cost-to-go'' 
function of Bellman's formulation of optimal control theory \cite{FlemingSoner,GuMo83}. 
If we exploit the anti-symmetry of the Poisson brackets 
\begin{eqnarray}
\label{Pontryagin:degeneracy1}
\brl S\,,\mathtt{m} \brr_{\mathsf{J}^{\dagger}}=-\brl \ln \mathtt{m}\,,\mathtt{m} \brr_{\mathsf{J}^{\dagger}}=0
\end{eqnarray}
and the definition of the non-equilibrium Helmholtz
energy density (\ref{tf:Hel}), we can always  couch $\mathfrak{L}^{\dagger}\mathtt{m}$ into
a first order differential operation over the probability density 
\begin{eqnarray}
\label{Pontryagin:generator}
\mathfrak{L}^{\dagger}\mathtt{m}=
-\partial_{\boldsymbol{x}}\cdot [\mathtt{m}\,(\mathsf{J}-\mathsf{G})\cdot\partial_{\boldsymbol{x}}A]
\end{eqnarray}
This fact accounts for regarding (\ref{Pontryagin:action}) as a functional of 
the non-equilibrium Helmholtz energy density $A$ and the probability density
$\mathtt{m}$. The right hand side of (\ref{Pontryagin:generator}) coincides with the 
$\mathbb{L}^{(2)}(\mathbb{R}^{2d})$-dual of the generator of deterministic transport 
by the vector field
\begin{eqnarray}
\label{Pontryagin:cgcv}
\boldsymbol{a}\equiv(\mathsf{J}-\mathsf{G})\cdot\partial_{\boldsymbol{x}}A
\end{eqnarray}
effectively describing a ``coarse graining'' of the underlying stochastic dynamics. 
Deterministic transport by (\ref{Pontryagin:cgcv}) arises from the fact that the entropy 
production is a functional of the \emph{individual probability density} specified by the boundary conditions
(\ref{Pontryagin:bc}).
This is at variance with the stochastic optimal control problems considered 
in \cite{FlemingRishel,FlemingSoner} where the cost or pay-off functional is a linear 
functional of the \emph{transition probability density} of the process. 
The entropy production optimal control problem belongs instead to the class encompassed by 
the ``weak-sense'' (stochastic) control theory of \cite{GuMo83}.

\subsection{Variations}
\label{sec:var}

We determine extremals of (\ref{Pontryagin:action}) by considering independent variations of 
$\mathtt{m}$, $V$ and $A$ in the admissible class (\ref{def:admissible}). The admissible class
hypothesis allows us to perform freely all the integrations by parts needed to extricate 
space-time local stationary conditions. After straightforward algebra (\ref{ap:var}), the variations of
$\mathtt{m}$, $V$ and $A$ respectively yield 
\begin{subequations}
\label{Occam:eqs}
\begin{eqnarray}
\label{Occam:value}
\tau\,\partial_{t}V+\brl A\,,V \brr_{\mathsf{J}^{\dagger}}
-\brl A\,,V \brr_{\mathsf{G}}
+\parallel\partial_{\boldsymbol{x}}A\parallel_{\mathsf{G}}^{2}=0
\end{eqnarray}
\begin{eqnarray}
\label{Occam:entropy}
\tau\,\partial_{t}S
+\brl A\,,S \brr_{\mathsf{J}^{\dagger}}+\frac{1}{\beta}\mathfrak{S}A=0
\end{eqnarray}
\begin{eqnarray}
\label{Occam:extremal}
\brl S\,,V \brr_{\mathsf{J}^{\dagger}}+\frac{1}{\beta}\mathfrak{S}V
=\frac{2}{\beta}\,\mathfrak{S}A
\end{eqnarray}
\end{subequations}
By $\mathfrak{S}$ we denote in (\ref{Occam:value}), (\ref{Occam:entropy}) the  operator 
\begin{eqnarray}
\label{var:WiLap}
\mathfrak{S}f=-\brl S,f\brr_{\mathsf{G}}+\mathsf{G}:\partial_{\boldsymbol{x}}\otimes\partial_{\boldsymbol{x}}f 
\end{eqnarray}
\emph{negative} definite for  any $f\in\mathbb{L}^{(2)}(\mathbb{R}^{2d},\mathtt{m}\,\mathrm{d}^{2d}x)$ (\ref{ap:var}).
The extremal equations (\ref{Occam:eqs}) are complemented by the boundary conditions:
\begin{eqnarray}
\label{Occam:bc}
S(\boldsymbol{x},\ti)=-\ln\frac{\mathtt{m}_{\mathrm{o}}(\boldsymbol{x})}{\beta^{d}}
\hspace{1.0cm}\&\hspace{1.0cm}
S(\boldsymbol{x},\tf)=-\ln\frac{\mathtt{m}_{\mathrm{f}}(\boldsymbol{x})}{\beta^{d}}
\end{eqnarray}
The value function (\ref{Occam:value}) and entropy (\ref{Occam:entropy}) equations 
describe deterministic transport by the ``coarse-grained'' current velocity (\ref{Pontryagin:cgcv}).
This latter vector field vanishes at equilibrium, so that (\ref{Occam:eqs}) in this case
admit the physically natural solution
\begin{eqnarray}
\label{}
\partial_{t}V=\partial_{t}S=A=0
\end{eqnarray}
with
\begin{eqnarray}
\label{}
H=\frac{1}{\beta}S
\end{eqnarray}
The condition (\ref{Occam:extremal}) plays for (\ref{Occam:eqs}) a role analogous to that
of pressure in hydrodynamics \cite{BlCrHoMa00}. It enforces a \emph{non-local coupling} between 
the microscopic entropy $S$, and the non-equilibrium Helmholtz energy density $A$. As in the case of 
hydrodynamics non-locality arises from the existence of a divergenceless component of the velocity field. 
In fact, neglecting the Poisson brackets in (\ref{Occam:extremal})
would allow us to recover the local extremal condition $V=2\,A$ analogous to the one minimizing the
entropy production by a Langevin--Smoluchowski dynamics \cite{AuMeMG12,AuGaMeMoMG12}.

Beside non-locality, a second major difference with Langevin--Smoluchowski is that the extremal 
equations (\ref{Occam:eqs}) are \emph{highly degenerate}. Namely, (\ref{Occam:extremal})
does not impose any constraint between the configuration space projection $\partial_{\boldsymbol{q}}A$ 
of the gradient of $A$ and the value function. This is a immediate consequence of the
independence of the entropy production from $\partial_{\boldsymbol{q}}A$. The generic consequence of degeneration 
is that (\ref{Occam:eqs}) describe a \emph{continuous family of controls} for which the entropy 
production attains a local, at least, minimum in $\mathbb{A}$. In the coming section~\ref{sec:ex}
we will illustrate the situation with an explicit example.

\section{An analytically solvable case}
\label{sec:ex}

We can explore more explicitly (\ref{Occam:eqs})
if we assume a Gaussian statistics for the initial and final states of the system. 
In particular, we restrict the attention to a two-dimensional phase space and 
suppose that the microscopic entropy of the initial $\mathrm{i}=\mathrm{o}$ and final
$\mathrm{i}=\mathrm{f}$ states be at most quadratic in $\boldsymbol{x}=q\oplus p$:
\begin{eqnarray}
\label{ex:Sfin}
\hspace{-0.8cm}
\lefteqn{
S_{\mathrm{i}}(p,q)
=
\frac{\beta\,(p-\mu_{p;\mathrm{i}})^{2}}{2\,\sigma_{p;\mathrm{i}}^{2}\,\cos^{2}\theta_{\mathrm{i}}}
+\frac{\beta\,(q-\mu_{q;\mathrm{i}})^{2}}{2\,\sigma_{q;\mathrm{i}}^{2}\,\cos^{2}\theta_{\mathrm{i}}}
}
\nonumber\\&&
\hspace{-0.6cm}
-\beta\,\tan\theta_{\mathrm{i}}\, \frac{(p-\mu_{p;\mathrm{i}})(q-\mu_{q;\mathrm{i}})}
{\sigma_{p;\mathrm{i}}\,\sigma_{q;\mathrm{i}}\,\cos\theta_{\mathrm{i}}}
-\ln\left(\frac{1}{2\,\pi\,\sigma_{p;\mathrm{i}}\,\sigma_{q;\mathrm{i}}\,\cos\theta_{\mathrm{i}}}\right)
\end{eqnarray}
corresponding to
\begin{eqnarray}
\label{}
\mathrm{E}\boldsymbol{\chi}_{t_{\mathrm{i}}}=\begin{bmatrix}
\mu_{q;\mathrm{i}}\\ \mu_{p;\mathrm{i}}
\end{bmatrix}
\end{eqnarray}
and
\begin{eqnarray}
\label{}
\mathrm{E}\left(\boldsymbol{\chi}_{t_{\mathrm{i}}}-\mathrm{E}\boldsymbol{\chi}_{t_{\mathrm{i}}}\right)\otimes
\left(\boldsymbol{\chi}_{t_{\mathrm{i}}}-\mathrm{E}\boldsymbol{\chi}_{t_{\mathrm{i}}}\right)
=\frac{1}{\beta}
\begin{bmatrix}
\sigma_{q;\mathrm{i}}^{2} & \sigma_{q;t_{\mathrm{i}}}\sigma_{p;t_{\mathrm{i}}}\sin\theta_{t_{\mathrm{i}}}
\\
\sigma_{q;\mathrm{f}}\sigma_{p;\mathrm{i}}\sin\theta_{t_{\mathrm{i}}} &\sigma_{p;t_{\mathrm{i}}}^{2}
\end{bmatrix}
\end{eqnarray}
In particular, we choose $\mu_{p;\mathrm{i}}=\mu_{q;\mathrm{i}}=\theta_{\mathrm{o}}=0$ 
whilst $0\,\leq\,\theta_{\mathrm{f}}\,<\pi/2$ parametrizes the 
degree of correlation between position and momentum variables of the final state. 
Under these assumptions, we look for the solution of the extremal equations
by means of quadratic Ans\"atze for the microscopic entropy
\begin{eqnarray}
\label{ex:S}
\hspace{-1.0cm}
\lefteqn{
S(p,q,t)=\frac{\beta\,(p-\mu_{p;t})^{2}}{2\,\sigma_{p;t}^{2}\,\cos^{2}\theta_{t}}
+\frac{\beta\,(q-\mu_{q;t})^{2}}{2\,\sigma_{q;t}^{2}\,\cos^{2}\theta_{t}}
}
\nonumber\\&&
\hspace{-0.8cm}
-\beta\,\tan\theta_{t} \frac{(p-\mu_{p;t})(q-\mu_{q;t})}{\sigma_{p;t}\,\sigma_{q;t}\,\cos\theta_{t}}
-\ln\frac{1}{2\,\pi\,\sigma_{p;t}\,\sigma_{q;t}\,\cos\theta_{t}}
\end{eqnarray}
and the non-equilibrium Helmholtz energy
\begin{eqnarray}
\label{ex:A}
A(p,q,t)=\frac{\mathsf{A}_{11;t}q^{2}+2\,\mathsf{A}_{12;t}p\,q+\mathsf{A}_{22;t}p^{2}}{2}
+a_{1;t}\,q+a_{2;t}\,p
\end{eqnarray}
for all $t\,\in\,[\ti,\tf]$. The Ans\"atze imply that the entropy production
\begin{eqnarray}
\label{ex:ep}
\lefteqn{
\mathcal{E}_{\tf,\ti}=
\beta\int_{\ti}^{\tf}\frac{\mathrm{d}t}{\tau}\,\left\{
2\,a_{2;t}\,\left(\mathsf{A}_{22;t}\,\mu_{p;t}+\mathsf{A}_{12;t}\,\mu_{q;t}\right)
+a_{2;t}^{2}
\right\}
}
\nonumber\\&&
+\beta\int_{\ti}^{\tf}\frac{\mathrm{d}t}{\tau}\,
\left\{\mathsf{A}_{22;t}\left(\mu_{p;t}^{2}+\sigma_{p;t}^{2}\,\cos^{2}\theta_{t}\right)
+\mathsf{A}_{12;t}\left(\mu_{q;t}^{2}+\sigma_{q;t}^{2}\,\cos^{2}\theta_{t}\right)\right\}
\nonumber\\&&
+2\,\beta\int_{\ti}^{\tf}\frac{\mathrm{d}t}{\tau}\,
\mathsf{A}_{12;t}\mathsf{A}_{22;t}\left(\mu_{p;t}\,\mu_{q;t}+\sigma_{p;t}\,\sigma_{q;t}\,\sin\theta_{t}\right)
\end{eqnarray}
does not depend explicitly upon $\mathsf{A}_{11;t}$ and $a_{1;t}$.

Using the quadratic Ans\"atze (\ref{ex:S}), (\ref{ex:A}) in (\ref{Occam:extremal}) 
we obtain
\begin{eqnarray}
\label{ex:V}
\hspace{-2.0cm}
V(p,q,t)=\left(
q^{2}\,\partial_{p}\partial_{q}-2\,q\,\partial_{p}\right)\left[A(p,q,t)
-\frac{y_{t}}{\beta}\,S(p,q,t)\right]
+\frac{2\,y_{t}}{\beta}\,S(p,q,t)+\bar{V}(t)
\end{eqnarray}
where 
\begin{eqnarray}
\label{ex:R}
y_{t}\equiv\beta\,\frac{\partial_{p}^{2}A}{\partial_{p}^{2}S}
\end{eqnarray}
is a function of the \emph{time variable alone} well-defined as long as the
probability density of the state is non-degenerate. The explicit value of 
$\bar{V}(t)$ does not play any role in the considerations which follow.
If we now insert (\ref{ex:V}) into (\ref{Occam:value}) and (\ref{Occam:entropy}), 
these equations foliate into a closed system of ordinary differential equations 
for the coefficients of the Ans\"atze (\ref{ex:S}) and (\ref{ex:A}). 
The calculation is laborious but straightforward. Upon setting 
\begin{eqnarray}
\label{ex:A22}
\mathsf{A}_{22:t}=-\frac{\tau\,\dot{y}_{t}}{y_{t}}
\end{eqnarray}
we find for the coefficients of second order monomials in (\ref{ex:S}) and (\ref{ex:A})
the set of relations
\begin{subequations}
\begin{eqnarray}
\label{ex:A11}
\mathsf{A}_{11:t}=\frac{y_{t}}{\beta}\,\left(\partial_{q}^{2}S-\partial_{p}\partial_{q}S
-\frac{y_{t}}{\dot{y}_{t}}\partial_{t}\partial_{p}\partial_{q}S\right)-\mathsf{A}_{12:t}
\end{eqnarray}
\begin{eqnarray}
\label{ex:A12}
\mathsf{A}_{12:t}=\frac{y_{t}}{\beta}\,\partial_{p}\partial_{q}S+\frac{\tau\,\dot{y}_{t}}{2\,y_{t}}
\end{eqnarray}
\begin{eqnarray}
\label{}
\partial_{p}^{2}S=-\frac{\tau\,\beta\,\dot{y}_{t}}{y_{t}^{2}}
\end{eqnarray}
\begin{eqnarray}
\label{}
\partial_{q}^{2}S=\frac{\left(\partial_{p}\partial_{q}S\right)^{2}}{\partial_{p}^{2}S}
-\tau\,\beta\,\frac{\ddot{y}_{t}^{2}-2\,\dot{y}_{t}\,\dddot{y}_{t}}{4\,\dot{y}_{t}^{3}}
\end{eqnarray}
\end{subequations}
The cross correlation coefficient $\partial_{p}\partial_{q}S$ of the microscopic entropy 
enters these equations as a free parameter only subject to the boundary conditions.
It turns out that the function $y_{t}$ must satisfy the fourth order non-linear differential 
equation
\begin{eqnarray}
\label{}
\ddddot{y}_{t}\,\dot{y}_{t}^{2}-2\,\dot{y}_{t}\,\ddot{y}_{t}\,\dddot{y}_{t}+\ddot{y}_{t}^{3}=0
\end{eqnarray}
with solution
\begin{eqnarray}
\label{ex:y}
y_{t}=\tau\,\Omega\,\left\{c_{0}
+c_{1}\,\Omega\,t+c_{1}\left[\sin\left(\Omega\,t+\varphi\right)-\sin\varphi\right]\right\}
\end{eqnarray}
The coefficients $c_{0},c_{1},\Omega,\varphi$ are fixed by the boundary 
conditions. Upon imposing the boundary conditions for
\begin{eqnarray}
\label{}
\ti=0
\end{eqnarray}
and requiring continuity of solutions for $S_{\mathrm{f}}\overset{\tf\downarrow 0}{\to}S_{o}$, we get into
\begin{eqnarray}
\label{ex:c0}
c_{0}=-\frac{\sigma_{p;\mathrm{o}}\,\sigma_{q;\mathrm{o}}}{2}
\end{eqnarray}
and
\begin{eqnarray}
\label{}
c_{1}=-\frac{\sigma_{q;\mathrm{o}}^{2}}{8\,\cos^{2}\frac{\varphi}{2}}
\end{eqnarray}
while $\Omega$ and $\varphi$ satisfy the transcendental equations:
\begin{subequations}
\label{ex:var_bc}
\begin{eqnarray}
\label{ex:sigmap}
\sigma_{p;\mathrm{f}}^{2}=\frac{\left\{4\,\sigma_{p;\mathrm{o}}\cos^{2}\frac{\varphi}{2}
+\sigma_{q;\mathrm{o}}\left[\Omega\,\tf+\sin(\Omega\,\tf+\varphi)-\sin\varphi\right]\right\}^{2}}
{16\,\cos^{2}\theta_{\mathrm{f}}\,\cos^{2}\frac{\varphi}{2}\,
\cos^{2}\frac{\Omega\,\tf+\varphi}{2}}
\end{eqnarray}
\begin{eqnarray}
\label{ex:sigmaq}
\frac{\sigma_{q;\mathrm{f}}^{2}}{\sigma_{q;\mathrm{o}}^{2}}=
\frac{\cos^{2}\frac{\Omega\,\tf+\varphi}{2}}{\cos^{2}\frac{\varphi}{2}}
\end{eqnarray}
\end{subequations}
We verify that the coefficients of the microscopic entropy are positive definite:
\begin{eqnarray}
\label{}
\hspace{-1.2cm}
\partial_{p}^{2}S=
\frac{16\,\beta\,\cos^{2}\frac{\varphi}{2}\,
\cos^{2}\frac{\Omega\,t+\varphi}{2}}
{\left\{4\,\sigma_{p;\mathrm{o}}\cos^{2}\frac{\varphi}{2}
+\sigma_{q;\mathrm{o}}\left[\Omega\,t+\sin(\Omega\,t+\varphi)-\sin\varphi\right]\right\}^{2}}
\geq\,0
\end{eqnarray}
and
\begin{eqnarray}
\label{}
\partial_{q}^{2}S=
\frac{\beta\,\cos^{2}\frac{\varphi}{2}}{\sigma_{q;\mathrm{o}}^{2}\,\cos^{2}\frac{\Omega\,t+\varphi}{2}}
+\frac{\left(\partial_{p}\partial_{q}S\right)^{2}}{\partial_{p}^{2}S}\,\geq\,0
\end{eqnarray}
The equations for the coefficients of the first degree monomials yield
\begin{eqnarray}
\label{}
\mu_{q:t}=\mu_{q;\mathrm{f}}\frac{t}{\tf}
\end{eqnarray}
and
\begin{eqnarray}
\label{}
a_{2;t}=\frac{\mu_{q;\mathrm{f}}\,\tau\,\left(2+\Omega\,t\,\tan\frac{\Omega\,t+\varphi}{2}\right)}{2\,\tf}
-\mu_{p;t}\mathsf{A}_{22;t}-\frac{y_{t}\,\mu_{q;t}}{\beta}\,\partial_{p}\partial_{q}S
\end{eqnarray} 
The remaining independent equations determine $a_{1;t}$ 
as a functional of $\partial_{p}\partial_{q}S$ and $\mu_{p;t}$ and their time derivatives. 
We do not need, however, the explicit expression to compute the entropy 
production for which we find
\begin{eqnarray}
\label{ex:ep1}
\frac{\mathcal{E}_{\hoz}}{\beta}=\frac{\mu_{q;\mathrm{f}}^{2}\,\tau}{\tf}
+\frac{\sigma_{q;\mathrm{o}}^{2}\,\Omega^{2}\,\tau\,\tf}{4\,\beta\,\cos^{2}\frac{\varphi}{2}}
\end{eqnarray}
Four properties of the extremal value of the entropy production (\ref{ex:ep1})
are worth emphasizing. First (\ref{ex:ep1}) is fully specified by the boundary conditions
and by the degrees of freedom fixed by the extremal equations (\ref{Occam:eqs}). This fact is 
an a-posteriori evidence of the degeneration of the extremal protocols.
Second, (\ref{ex:ep1}) does not depend upon the expected value of the momentum variable but only 
upon its variance. The third property is that, (\ref{ex:ep1}) corresponds 
to a \emph{constant} entropy production rate over the transition horizon.
This phenomenon is reminiscent of the Langevin--Smoluchowski case where the entropy production 
coincides with the kinetic energy of the current velocity so that the optimal value is 
attained along free--streaming trajectories.
The fourth property pertains the limit of infinite time horizon $\tf\uparrow\infty$.
The position variable variance remains finite in such a limit if $\Omega\,\tf$ is finite.
This lead us to infer generically a $1/\tf$ decay of the entropy production in such limit.
  
The explicit dependence of (\ref{ex:ep1}) on the boundary conditions
can be obtained in several special cases.

\subsection{$\sigma_{q;\mathrm{o}}=\sigma_{q;\mathrm{f}}$ }
\label{sec:ex:azero}

The condition is satisfied for
\begin{eqnarray}
\label{ex:azero:y}
\Omega=y_{t}=0
\end{eqnarray}
in the control horizon. Correspondingly, (\ref{ex:sigmap}) yields
\begin{eqnarray}
\label{ex:azero:sigmap}
\sigma_{p;\mathrm{f}}=\frac{\sigma_{p;\mathrm{o}}}
{\cos\theta_{\mathrm{f}}}
\end{eqnarray}
If $\theta_{\mathrm{f}}\neq 0$, (\ref{ex:azero:sigmap}) states that, while enforcing
(\ref{ex:azero:y}) we can use $\partial_{p}\partial_{q}S$ to steer the system to a final state
with larger momentum variance and non-vanishing correlation between position
and momenta.
For vanishing $\Omega$ the entropy production is determined by the variation 
of the position average:
\begin{eqnarray}
\label{}
\frac{\mathcal{E}_{\tf,0}}{\beta}=\frac{\mu_{q;\mathrm{f}}^{2}\,\tau}{\tf}
\end{eqnarray}
Correspondingly, the non-equilibrium Helmholtz energy and the stochastic
entropy can be couched into the form
\begin{subequations}
\label{ex:azero:sol}
\begin{eqnarray}
\label{azero:A}
\hspace{-1.4cm}
\lefteqn{
A=\frac{\mu_{q;\mathrm{f}}\,\tau\,p}{\tf}+
}
\nonumber\\&&
\frac{\sigma_{p;\mathrm{o}}\,q^{2}}{2\,\sigma_{q;\mathrm{o}}}\,\tau\,
\dfrac{\mathrm{d}}{\mathrm{d}t}\tanh\theta_{t}
-\frac{\tau\,q}{\tf}\left(\mu_{q;\mathrm{f}}
+\tf\,\dot{\mu}_{p;t}+\frac{\mu_{q;\mathrm{f}}\,\sigma_{p;\mathrm{o}}}{\sigma_{q;\mathrm{o}}}\,t\,
\frac{\mathrm{d}}{\mathrm{d} t}\tanh\theta_{t}\right)
\end{eqnarray}
\begin{eqnarray}
\label{azero:S}
\hspace{-1.4cm}
\lefteqn{
S=\frac{\beta\,(p-\mu_{p;t})^{2}}{2\,\sigma_{p;\mathrm{o}}^{2}}
+\frac{\beta}{2\,\sigma_{q;\mathrm{o}}^{2}\,\cos^{2}\theta_{t}}\left(q-\frac{\mu_{q;\mathrm{f}}\,t}{\tf}\right)^{2}
}
\nonumber\\&&
\hspace{-0.6cm}
-\beta\frac{\tanh\theta_{t}}{\sigma_{p;\mathrm{o}}\,\sigma_{q;\mathrm{o}}} 
\left(p-\mu_{p;t}\right)\left(q-\frac{\mu_{q;\mathrm{f}}\,t}{\tf}\right)
-\ln\frac{1}{2\,\pi\,\sigma_{p;\mathrm{o}}\,\sigma_{q;\mathrm{o}}}
\end{eqnarray}
\end{subequations} 
with $\tanh\theta_{t}$, $\sigma_{q;t}$ and $\dot{\mu}_{p;t}$ \emph{arbitrary} differentiable functions matching the boundary conditions. 
If we add the requirement $\theta_{\mathrm{f}}=\theta_{t}=0$ (\ref{ex:azero:sol}) shows that a transition changing only the 
mean value of the position variable requires a quadratic additive Hamiltonian i.e. of the form
 $H(p,q,t)=H_{p}(p,t)+H_{q}(q,t)$ where $H_{p}(p,t)$ must, however, include a 
\emph{linear momentum dependence}.

\subsection{Small one-parameter variation of the statistics}
\label{sec:ex:eps}

Let us suppose that there exists an non-dimensional parameter $\varepsilon$
such that the elements of the correlation matrix of the final state admit
an expansion of the form
\begin{subequations}
\begin{eqnarray}
\label{}
\sigma_{p;\mathrm{f}}=\sigma_{p;\mathrm{o}}
\left[1+\mathsf{p}_{1}\varepsilon+\frac{\mathsf{p}_{2}\varepsilon^{2}}{2}
+\frac{\mathsf{p}_{3}\varepsilon^{3}}{6}+O(\varepsilon^4)\right]
\end{eqnarray}
\begin{eqnarray}
\label{}
\sigma_{q;\mathrm{f}}=\sigma_{q;\mathrm{o}}\left[1
+\mathsf{q}_{1}\varepsilon+\frac{\mathsf{q}_{2}\varepsilon^{2}}{2}
+\frac{\mathsf{q}_{3}\varepsilon^{3}}{6}+O(\varepsilon^4)\right]
\end{eqnarray}
\end{subequations}
with
\begin{eqnarray}
\label{}
\cos\theta_{\mathrm{f}}=1-\frac{w_{2}^{2}\,\varepsilon^{2}}{2}+O(\varepsilon^{4})
\end{eqnarray}
Under the foregoing hypothesis we find 
\begin{subequations}
\label{ex:eps:phafre}
\begin{eqnarray}
\label{ex:eps:omega}
\hspace{-1.6cm}
\Omega=\frac{2\,(\mathsf{q}_{1}+\mathsf{p}_{1})\,\sigma_{p;\mathrm{o}}\,\varepsilon}{\tf\,\sigma_{q;\mathrm{o}}}+
\frac{2\,\left(\mathsf{p}_{2}-w_{2}^{2}+\mathsf{q}_{2}-2\,\mathsf{q}_{1}^{2}\right)\,\sigma_{p;\mathrm{o}}\varepsilon^{2}}
{\tf\,\sigma_{q;\mathrm{o}}}+O(\varepsilon^{3})
\end{eqnarray}
\begin{eqnarray}
\label{ex:eps:varphi}
\hspace{-1.6cm}
\lefteqn{
\varphi=-2\,\mathrm{arccot}\frac{(\mathsf{q}_{1}+\mathsf{p}_{1})\,\sigma_{p;\mathrm{o}}}{\mathsf{q}_{1}\,\sigma_{q;\mathrm{o}}}
}
\nonumber\\&&
\hspace{-1.4cm}
-\frac{(\mathsf{q}_{1}+\mathsf{p}_{1})^{3}\,\sigma_{p;\mathrm{o}}^{2}
+\left[\mathsf{q}_{1}(w_{2}^{2}+2\,\mathsf{q}_{1}^{2})
+\mathsf{p}_{1}\mathsf{q}_{2}-\mathsf{p}_{2}\mathsf{q}_{1}\right]\sigma_{q;\mathrm{o}}^{2}}
{\left[(\mathsf{q}_{1}+\mathsf{p}_{1})^{2}\sigma_{q;\mathrm{o}}^{2}
+\mathsf{q}_{1}^{2}\sigma_{q;\mathrm{o}}^{2}\right]\sigma_{q;\mathrm{o}}}
\,\sigma_{p;\mathrm{o}}\,\varepsilon+O(\varepsilon^{2})
\end{eqnarray}
\end{subequations}
which give for the entropy production
\begin{eqnarray}
\label{}
\hspace{-1.6cm}
\lefteqn{
\frac{\mathcal{E}_{\tf,0}}{\beta}=\frac{\mu_{q;\mathrm{f}}^{2}\,\tau}{\tf}+
\frac{\left[\left(\mathsf{p}_{1}+\mathsf{q}_{1}\right)^{2}\,\sigma_{p;\mathrm{o}}^{2}
+\mathsf{q}_{1}^{2}\,\sigma_{q;\mathrm{o}}^{2}\right]\,\tau\,\varepsilon^{2}}{\beta\,\tf}
}
\nonumber\\&&
\hspace{-1.4cm}
+\frac{\left[(\mathsf{p}_{1}+\mathsf{q}_{1})\left(\mathsf{p}_{2}
+\mathsf{q}_{2}+(\mathsf{p}_{1}-\mathsf{q}_{1})\,\mathsf{q}_{1}-w_{2}^{2}\right)\,\sigma_{p;\mathrm{o}}^{2}
+\mathsf{q}_{2}\,\mathsf{q}_{1}\sigma_{q;\mathrm{o}}^{2}
\right]
\,\tau\,\varepsilon^{3}}{\beta\,\tf}+O(\varepsilon^{4})
\end{eqnarray}
It is interesting to explore the consequence of this formula in three sub-cases.
For this purpose we introduce the non-dimensional parameter
\begin{eqnarray}
\label{ex:lambda}
\lambda=\frac{\sigma_{p;\mathrm{o}}}{\sigma_{q;\mathrm{o}}}
\end{eqnarray}
measuring the scale separation between momentum and position fluctuations.

\subsubsection{$\mathsf{q}_{n}=\mathsf{p}_{n}=0$ $\forall\,n\,>\,1$}~
\label{sec:ex:sp}

Both the position and the momentum variances are linear in $\varepsilon$. 
We can therefore recast the expansion of the entropy production directly
in terms of the change of the variances across the control horizon.
We obtain 
\begin{eqnarray}
\label{}
\hspace{-2.0cm}
\lefteqn{
\frac{\mathcal{E}_{\tf,0}}{\beta}=\frac{\mu_{q;\mathrm{f}}^{2}\,\tau}{\tf}
+\frac{(\sigma_{q;\mathrm{f}}-\sigma_{q;\mathrm{o}})^{2}\,\tau}{\beta\,\tf}
}
\nonumber\\
\hspace{-1.8cm}
+\frac{[\sigma_{p;\mathrm{f}}-\sigma_{p;\mathrm{o}}
+\lambda\,(\sigma_{q;\mathrm{f}}-\sigma_{q;\mathrm{o}})]^{2}\tau}{\beta\,\tf}
+\frac{[(\sigma_{p;\mathrm{f}}-\sigma_{p;\mathrm{o}})^{2}-\lambda^{2}\,(\sigma_{q;\mathrm{f}}-\sigma_{q;\mathrm{o}})^{2}]
(\sigma_{q;\mathrm{f}}-\sigma_{q;\mathrm{o}})
\,\tau}{\beta\,\sigma_{q;\mathrm{o}}\,\tf}
\nonumber\\
\hspace{-1.8cm}
+\frac{2\,\lambda\,[\sigma_{p;\mathrm{f}}-\sigma_{p;\mathrm{o}}+\lambda\,(\sigma_{q;\mathrm{f}}-\sigma_{q;\mathrm{o}})]
\left(1-\cos\theta_{\mathrm{f}}\right)\,\tau}{\beta\,\sigma_{q;\mathrm{o}}\,\tf}
+h.o.t
\end{eqnarray}

\subsubsection{$\theta_{\mathrm{f}}=\sigma_{p;\mathrm{f}}-\sigma_{p;\mathrm{o}}=0$ and $\varepsilon=(\sigma_{q;\mathrm{f}}-\sigma_{q;\mathrm{o}})/\sigma_{q;\mathrm{o}}$} ~
\label{sec:ex:sq}

Under these hypotheses, the marginal momentum distribution in the final state 
coincides with that of the initial state. As a result, the expansion of the phase 
$\varphi$ starts from the neighborhood of $\pi/2$. The entropy production reduces to
\begin{eqnarray}
\label{}
\hspace{-1.6cm}
\frac{\mathcal{E}_{\tf,0}}{\beta}=\frac{\mu_{q;\mathrm{f}}^{2}\,\tau}{\tf}
+\frac{\tau\,(1+\lambda^{2})\,(\sigma_{q;\mathrm{f}}-\sigma_{q;\mathrm{f}})^{2}}{\beta\,\tf}
-\frac{\tau\,\lambda^{2}\,(\sigma_{q;\mathrm{f}}-\sigma_{q;\mathrm{o}})^{3}}{\beta\,\sigma_{q;\mathrm{o}}\,\tf}
+O(\sigma_{q;\mathrm{f}}-\sigma_{q;\mathrm{o}})^{4}
\end{eqnarray}

\begin{figure}[ht]
\centering
\subfigure[Momentum variance $\sigma_{p;t}^{2}$]{
\includegraphics[width=4.2cm]{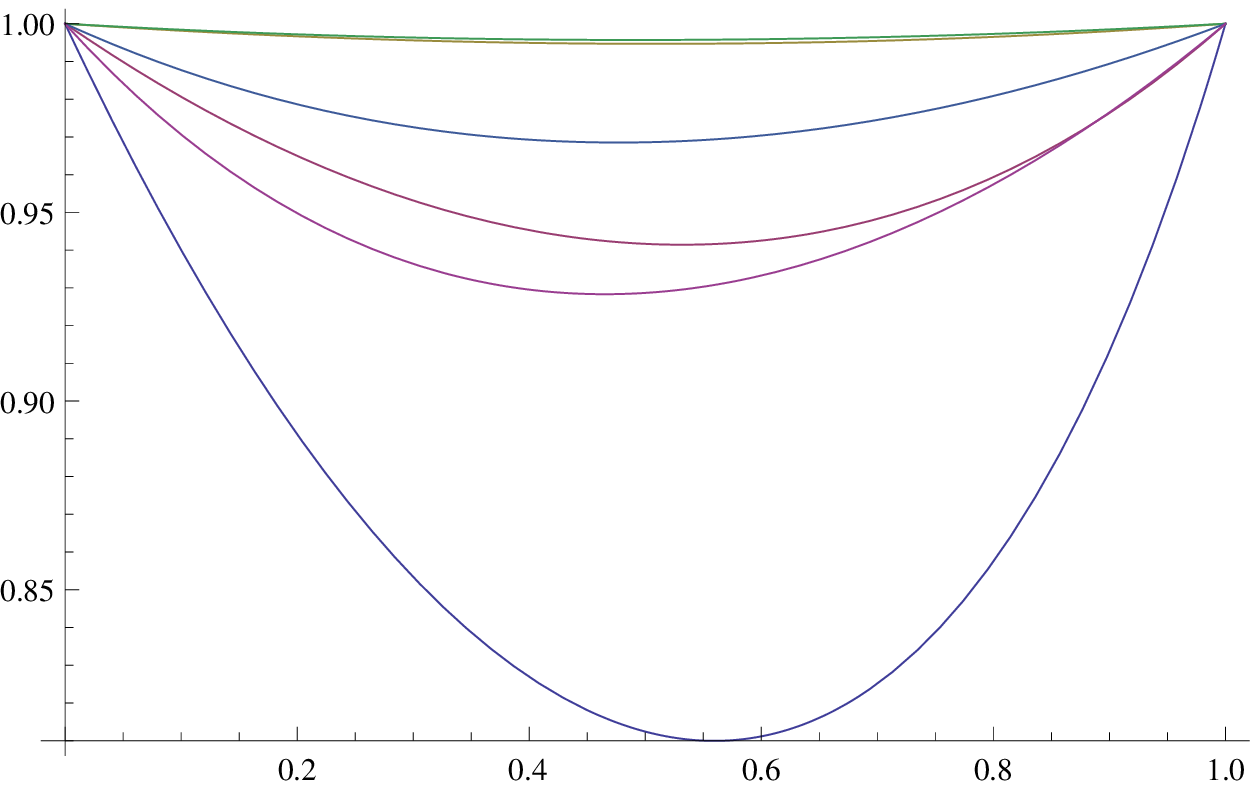}
\label{fig:var_mom}
}
\hspace{0.8cm}
\subfigure[Position variance $\sigma_{q;t}^{2}$ for $\partial_{q}\partial_{p}S=0$ ]{
\includegraphics[width=4.2cm]{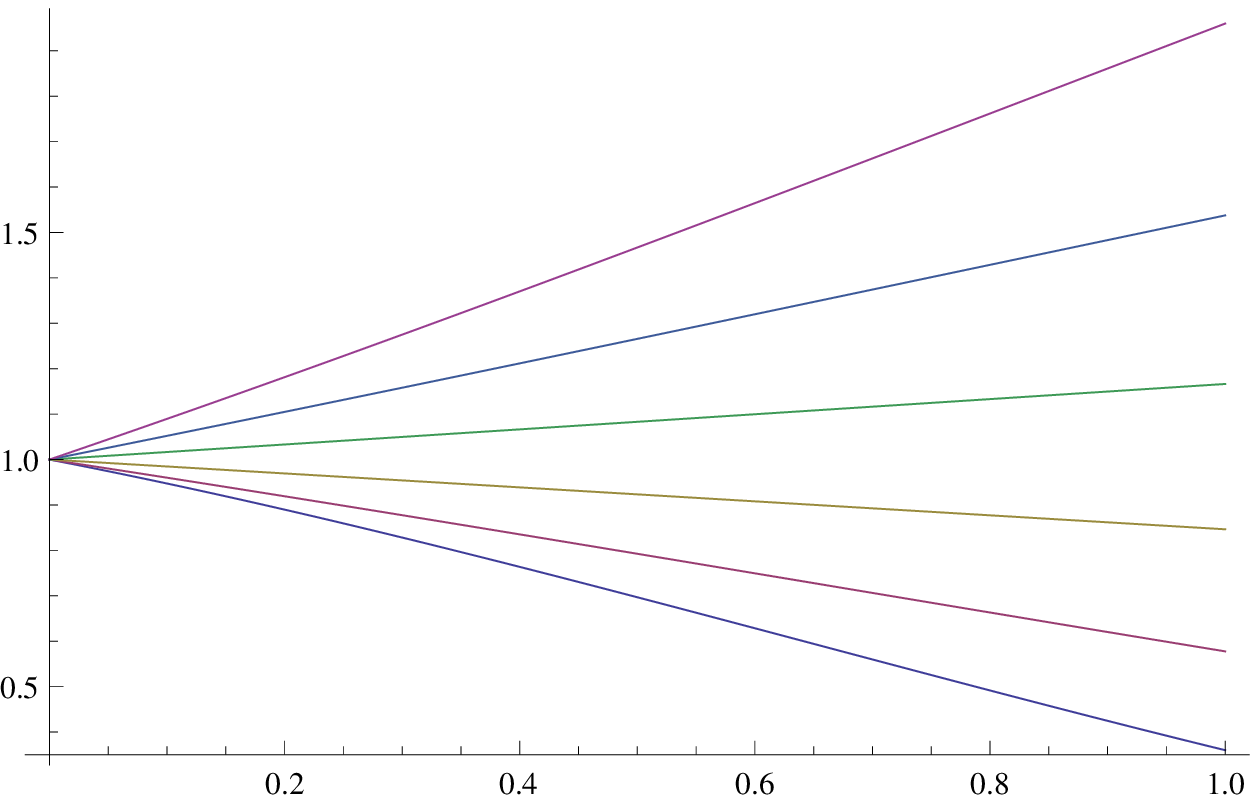}
  \label{fig:var_pos}
}
\hspace{0.8cm}
\subfigure[Contribution to the entropy production by a change of position variance]{
  \includegraphics[width=4.2cm]{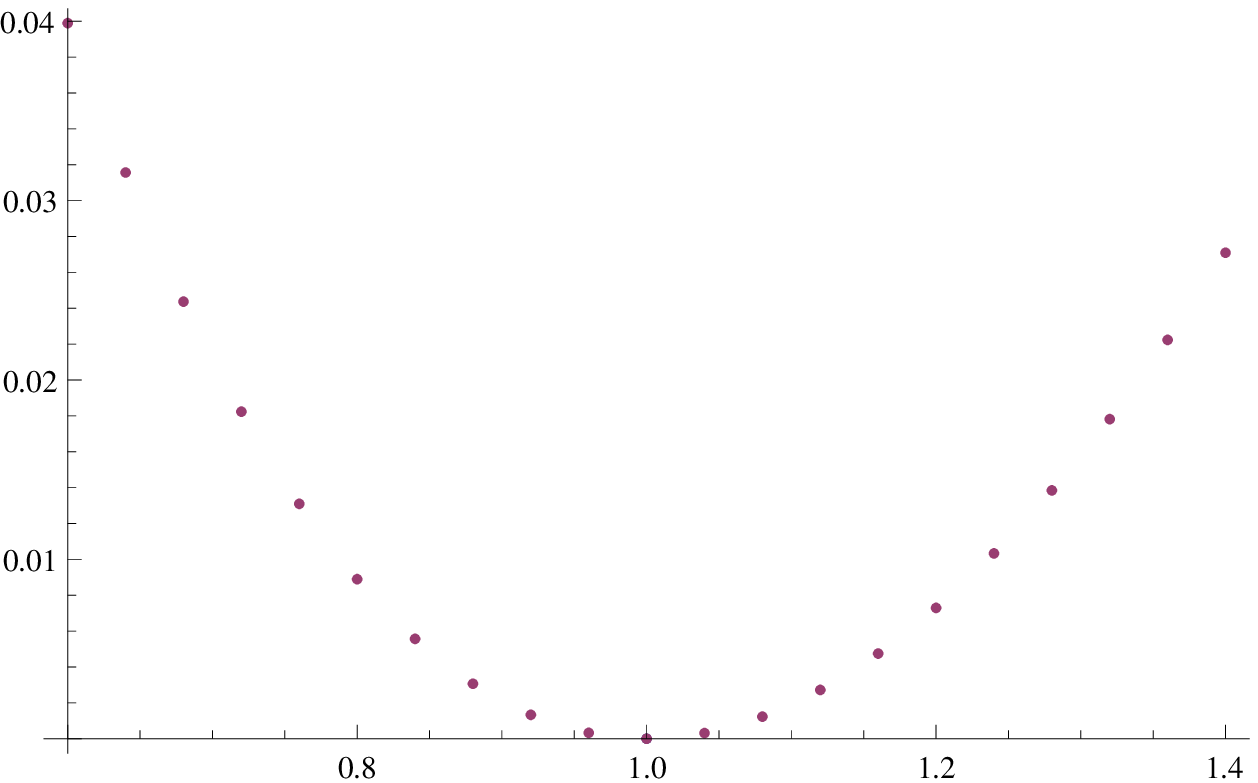}
  \label{fig:entropy}
}
\setlength{\unitlength}{0.1cm}
\begin{picture}{(0,0)}
\put(-136,5){$\scriptscriptstyle{\delta=-0.4}$}
\put(-136,19){$\scriptscriptstyle{\delta=-0.24}$}
\put(-148,26){$\scriptscriptstyle{\delta=-0.08}$}
\put(-122,26){$\scriptscriptstyle{\delta=0.08}$}
\put(-136,22){$\scriptscriptstyle{\delta=0.24}$}
\put(-136,16){$\scriptscriptstyle{\delta=0.4}$}
\put(-55,1){$\scriptscriptstyle{\delta=-0.4}$}
\put(-55,4){$\scriptscriptstyle{\delta=-0.24}$}
\put(-55,8){$\scriptscriptstyle{\delta=-0.08}$}
\put(-55,14){$\scriptscriptstyle{\delta=0.08}$}
\put(-55,20){$\scriptscriptstyle{\delta=0.24}$}
\put(-55,26){$\scriptscriptstyle{\delta=0.4}$}
\end{picture}
\caption{Evolution of the fluctuation variances for $\tf=1$ and $\tau=0.1$. As in 
subsection~\ref{sec:ex:sq} we assume $\sigma_{p;\mathrm{o}}=\sigma_{p;\mathrm{f}}=\sigma_{q;\mathrm{o}}=1$ and 
$\theta_{\mathrm{f}}=0$. Correspondingly, we plot fig.~\ref{fig:var_pos} after 
imposing $\partial_{q}\partial_{p}S=0$ in the control horizon. The plots correspond to different 
values of $\delta=\sigma_{q;\mathrm{f}}-1$.
In fig~\ref{fig:var_mom} the constant value of $\sigma_{p;t}$ corresponds to $\sigma_{q;\mathrm{f}}=1.0$. 
The larger the absolute value of the deviation 
from unity of $\sigma_{q;\mathrm{f}}$, the larger the departure of $\sigma_{p;t}$ from the unity \emph{inside} 
the control horizon. In fig~\ref{fig:entropy} we plot the corresponding 
values of $\mathcal{E}_{\tf,0}-\beta \mu_{q;\mathrm{f}}^{2}\tau/\tf$
}
\label{fig:ex:sp}
\end{figure}
In fig.~\ref{fig:ex:sp} we report  the behavior of the momentum and position variance 
for vanishing cross-correlation. It is worth emphasizing that the momentum variance 
does not remain constant during the control horizon unless
\begin{eqnarray}
\label{}
\sigma_{q;\mathrm{f}}=\sigma_{q;\mathrm{o}}=\Omega=0
\end{eqnarray}

\subsubsection{``Over-damped'' limit: 
$\theta_{\mathrm{f}}=\sigma_{p;\mathrm{f}}-\sigma_{p;\mathrm{o}}=0$ and 
$\varepsilon=(\sigma_{q;\mathrm{f}}-\sigma_{q;\mathrm{o}})/\sigma_{q;\mathrm{o}}$ for $\lambda\,\ll \,1$}~
\label{sec:ex:od}

At variance with the foregoing~\label{sec:ex:sq} we now assume a wide scale separation 
between the position and momentum variance. We readily see from  (\ref{ex:eps:phafre}) 
that $(\varphi\,,\Omega)=(-\pi+O(\lambda)\,,O(\lambda))$. We can solve the boundary condition equations 
(\ref{ex:var_bc}) in the limit of vanishing 
$\lambda$ up to all order accuracy in $\varepsilon$:
\begin{subequations}
\label{ex:od:var}
\begin{eqnarray}
\label{}
\Omega=\frac{2\,\lambda\,\varepsilon}{\tf}\left\{1-\varepsilon+\frac{2\,\varepsilon^{2}}{3}+O(\varepsilon^{3})\right\}
\overset{\lambda\downarrow 0}{\to}\frac{6 \,\lambda\, \varepsilon}
{\tf\,\left[3+\varepsilon\,(3+\varepsilon)\right]}+o(\lambda)
\end{eqnarray}
\begin{eqnarray}
\label{}
\varphi=-\pi+2\,\lambda\left\{1-\varepsilon+\frac{2\,\varepsilon^{2}}{3}+O(\varepsilon^{3})\right\}
\overset{\lambda \downarrow 0}{\to}-\pi+\frac{\Omega\,\tf}{\varepsilon}+o(\lambda)
\end{eqnarray}
\end{subequations}
The corresponding value of the entropy production is
\begin{eqnarray}
\label{ex:od:entropy}
\frac{\mathcal{E}_{\tf,0}}{\beta}=\frac{\mu_{q;\mathrm{f}}^{2}\,\tau}{\tf}
+\frac{(\sigma_{q;\mathrm{f}}-\sigma_{q;\mathrm{o}})^{2}}{\beta\,\tf}+
o\left(\lambda\right)
\end{eqnarray}
We notice that this is exactly the entropy production by a transition governed
by a Langevin--Smoluchowski dynamics between Gaussian states \cite{AuMeMG12,AuGaMeMoMG12,PMG13}. 
Indeed, the regime we are considering here corresponds to the ``\emph{over-damped}'' asymptotics 
of the Langevin--Kramers dynamics. 
Namely, upon inserting (\ref{ex:od:var}) into the quadratic Ans\"atze for the
non-equilibrium Helmholtz energy and microscopic entropy densities we get into
\begin{subequations}
\begin{eqnarray}
\label{ex:od:Aderq}
\left.(\partial_{q}A)(0,q,t)\right|_{\mu_{p;t}=0}=
-\frac{\mu_{q;\mathrm{f}}\,+\frac{q\,(\sigma_{q;\mathrm{f}}-\sigma_{q;\mathrm{o}})}{\sigma_{q;\mathrm{o}}}}
{1+\frac{t\,\left(\sigma_{q;\mathrm{f}}-\sigma_{q;\mathrm{o}}\right)}{\tf\,\sigma_{q;\mathrm{o}}}}\frac{\tau}{\tf}
+o\left(\lambda\right)
\end{eqnarray}
\begin{eqnarray}
\label{ex:od:Aderp}
\left.(\partial_{p}A)(0,q,t)\right|_{\mu_{p;t}=0}=-\left.(\partial_{p}A)(0,q,t)\right|_{\mu_{p;t}=0}+o\left(\lambda\right)
\end{eqnarray}
\begin{eqnarray}
\label{ex:od:Sderq}
(\partial_{q}S)(0,q,t)=\frac{\beta\,\left(q-\frac{\mu_{q;\mathrm{f}}\,t}{\tf}\right)}
{\sigma_{q;\mathrm{o}}^{2}\,\left[1+\frac{t\,(\sigma_{q;\mathrm{f}}-\sigma_{q;\mathrm{o}})}{\tf\,\sigma_{q;\mathrm{o}}}\right]^{2}}
+o(\lambda)
\end{eqnarray}
\end{subequations}
We see that the $\lambda$-independent parts of 
(\ref{ex:od:Aderq}) and (\ref{ex:od:Sderq}) coincide with the values obtained
for the same quantities in the Langevin--Smoluchowski case \cite{AuMeMG12,AuGaMeMoMG12,PMG13}. 
In the forthcoming section, we will show that (\ref{Occam:eqs}) encapsulate also in general the results 
obtained for Langevin--Smoluchowski dynamics. 
In particular, the equality (\ref{ex:od:Aderp}) guarantees that an homogenization theory ``centering condition''
holds for the Gaussian model so that the Langevin--Smoluchowski dynamics
is recovered as the solution of a suitable ``cell problem'' \cite{PaSt08}.
\begin{figure}[ht]
\centering
\subfigure[Relative variation of the momentum variance for $\sigma_{p;\mathrm{o}}=\lambda=1$]{
\includegraphics[width=5.0cm]{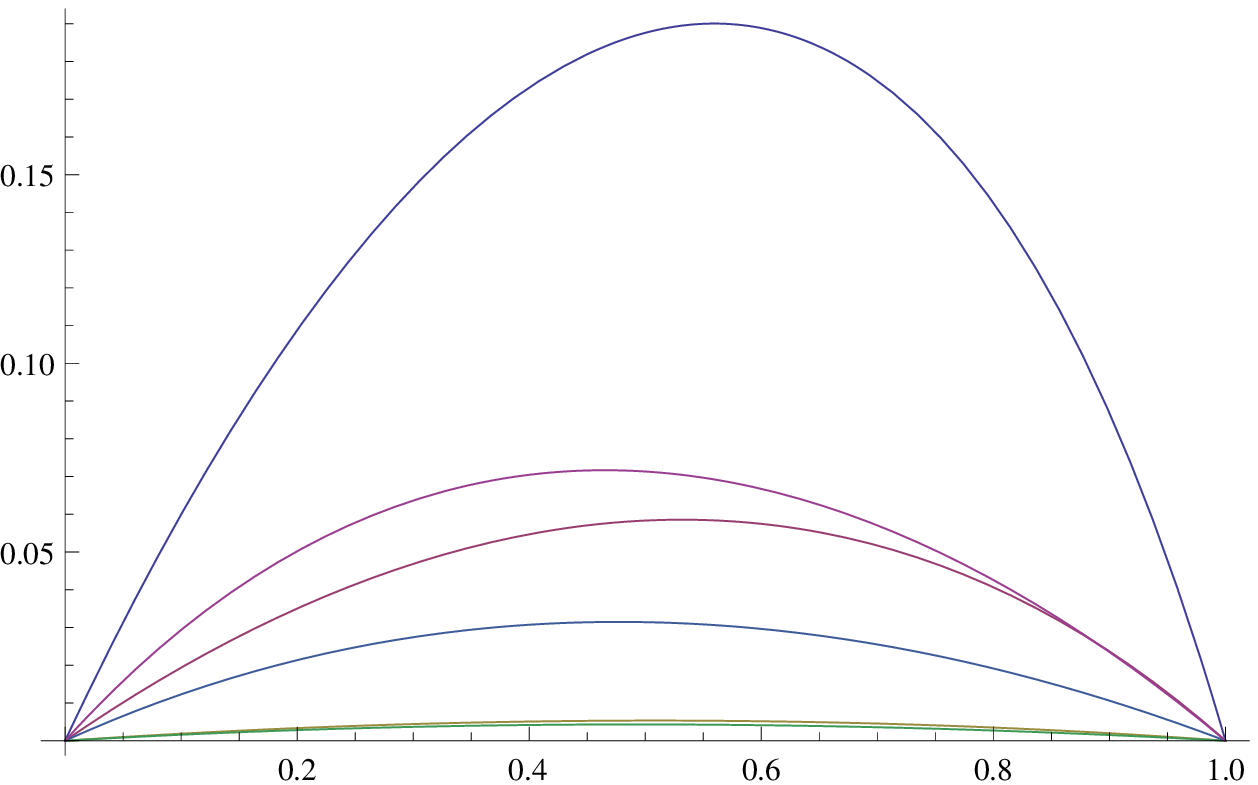}
\label{fig:relvar_sq}
}
\hspace{1.0cm}
\subfigure[Relative variation of the momentum variance for $\sigma_{p;\mathrm{o}}=\lambda=0.1$ ]{
\includegraphics[width=5.0cm]{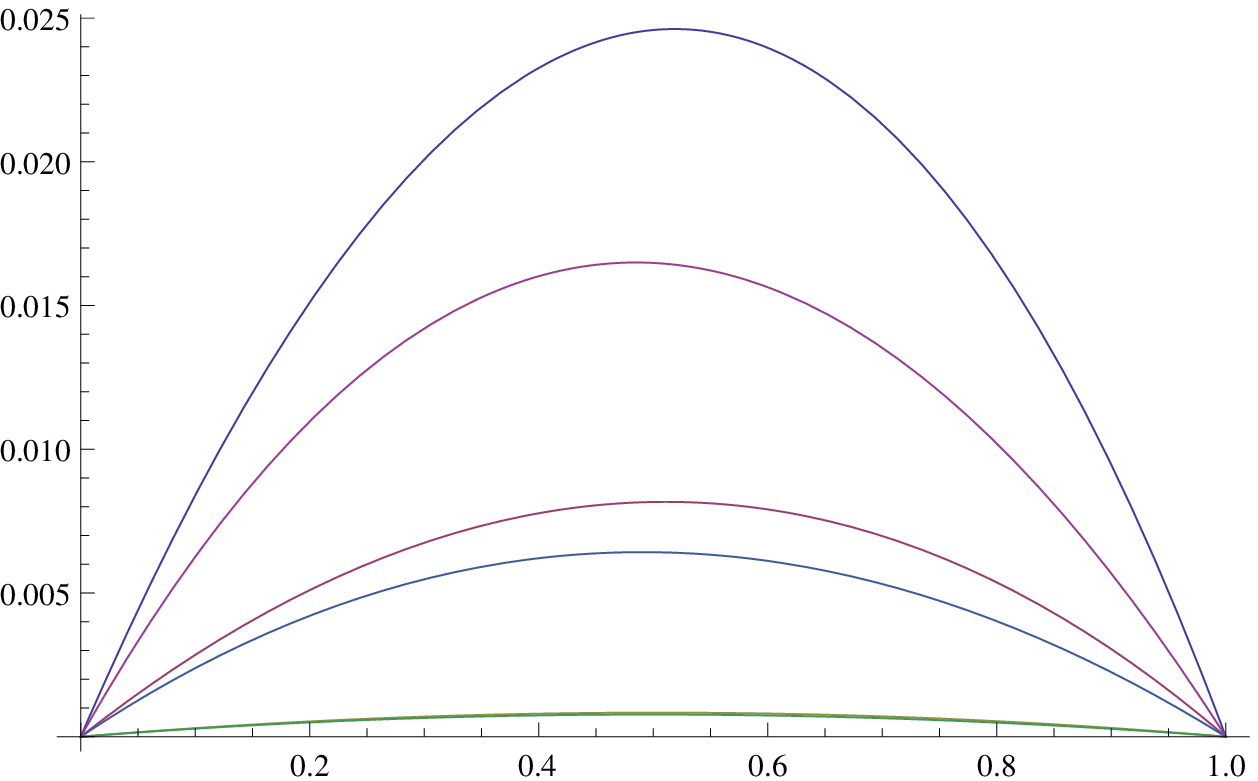}
  \label{fig:relvar_od}
}
\setlength{\unitlength}{0.1cm}
\begin{picture}{(0,0)}
\put(-94,26){$\scriptscriptstyle{\delta=-0.4}$}
\put(-94,9){$\scriptscriptstyle{\delta=-0.24}$}
\put(-108,3){$\scriptscriptstyle{\delta=-0.08}$}
\put(-82,3){$\scriptscriptstyle{\delta=0.08}$}
\put(-94,5){$\scriptscriptstyle{\delta=0.24}$}
\put(-94,14){$\scriptscriptstyle{\delta=0.4}$}
\put(-30,26){$\scriptscriptstyle{\delta=-0.4}$}
\put(-30,9){$\scriptscriptstyle{\delta=-0.24}$}
\put(-44,3){$\scriptscriptstyle{\delta=-0.08}$}
\put(-18,3){$\scriptscriptstyle{\delta=0.08}$}
\put(-30,5){$\scriptscriptstyle{\delta=0.24}$}
\put(-30,14){$\scriptscriptstyle{\delta=0.4}$}
\end{picture}
\caption{Relative variation $|\sigma_{p;t}^{2}-\sigma_{p;\mathrm{o}}^{2}|/\sigma_{p;\mathrm{o}}^{2}$ 
for $\sigma_{p;\mathrm{o}}=1$ fig.~\ref{fig:relvar_sq} and $\sigma_{p;\mathrm{o}}=0.1$ fig.~\ref{fig:relvar_sq}.
In both cases $\sigma_{q;\mathrm{o}}=1$ so that $\sigma_{p;\mathrm{o}}=\lambda$.
The other parameters as in fig.~\ref{fig:ex:sq}. As $\lambda$ decreases, the
approximation of the marginal momentum distribution by an ``equilibrium'' distribution improves its accuracy
while remaining not uniform inasmuch the deviation increases with $(\sigma_{q;\mathrm{f}}-\sigma_{q;\mathrm{o}})^{2}$
and reaches a maximum for $t\sim \tf/2$. 
}
\label{fig:ex:sq}
\end{figure}

\section{``Over-damped'' asymptotics}
\label{sec:multi}

In the presence of a wide separation of between the characteristic scales of the 
momentum and position variables, the Langevin--Smoluchowski or ``over-damped'' dynamics often
provides a good approximation to the Langevin--Kramers dynamics.
The entropy production by a smooth Langevin--Smoluchowski dynamics attains a minimum value
if the control potential obeys a Monge--Amp\`ere--Kantorovich dynamics \cite{AuMeMG11,AuMeMG12,AuGaMeMoMG12,PMG13}.
In this last section our aim is to investigate in which sense we can recover from
(\ref{Occam:eqs}) the results
previously established for Langevin-Smoluchowski dynamics. To address this question, we suppose 
that the probability densities of the initial and final states take the additive form
\begin{eqnarray}
\label{multi:ini}
\mathtt{m}_{\mathrm{o}}(\boldsymbol{p},\boldsymbol{q})=\left(\frac{\beta}{2\,\pi\,\lambda}\right)^{d}
\exp\left\{-\beta\,\frac{\parallel\boldsymbol{p}\parallel_{\id_{d}}^{2}}{2\,\lambda^{2}}
-\beta\,U_{\mathrm{o}}(\boldsymbol{q})\right\}
\end{eqnarray}
and
\begin{eqnarray}
\label{multi:fin}
\mathtt{m}_{\mathrm{f}}(\boldsymbol{p},\boldsymbol{q})=\left(\frac{\beta}{2\,\pi\,\lambda}\right)^{d}
\exp\left\{-\beta\,\frac{\parallel\boldsymbol{p}\parallel_{\id_{d}}^{2}}{2\,\lambda^{2}}
-\beta\,U_{\mathrm{f}}(\boldsymbol{q})\right\}
\end{eqnarray}
with $\lambda\,\ll\,1$ a non-dimensional parameter generalizing (\ref{ex:lambda}) in order to describe 
the scale separation between momentum and position dynamics.

Multi-scale perturbation theory (often also referred to as homogenization theory 
see e.g. \cite{BeLiPa78,PaSt08}) in powers of $\lambda$ equips us with the tools to extricate 
the \emph{asymptotic expression} of solutions of (\ref{Occam:eqs}) for $\beta\,\parallel\boldsymbol{p}\parallel_{\id_{d}}^{2}\ll\lambda\ll\,1$ 
in the form
\begin{eqnarray}
\label{multi:expansion}
\hspace{-1.6cm}
A(\boldsymbol{x},t)=\sum_{i=0}^{2}\,\lambda^{i}\,
A_{(i)}\left(\frac{\boldsymbol{p}}{\lambda},\boldsymbol{q},t\dots\right)
+o(\lambda^{2}):=\tilde{A}\left(\tilde{\boldsymbol{p}},\boldsymbol{q},t\dots\right)
\end{eqnarray}
and similarly for $S$ and $V$.  The $\dots$ in (\ref{multi:expansion}) portend the scales
which we eventually neglect in the asymptotics.     
Once we availed us of (\ref{multi:expansion}), the extremal equations (\ref{Occam:eqs}) 
become
\begin{subequations}
\label{multi:eqs}
\begin{eqnarray}
\label{multi:extremal}
\hspace{-1.2cm}
\frac{1}{\lambda}\brl \tilde{S}\,,\tilde{V}\brr_{\mathsf{J}^{\dagger}}^{\sim}
+\frac{1}{\lambda^{2}\,\beta}\tilde{\mathfrak{S}}(\tilde{V}-2\,\tilde{A})
=0
\end{eqnarray}
\begin{eqnarray}
\label{multi:value}
\hspace{-1.2cm}
\tau\,\partial_{t}\tilde{V}+\frac{1}{\lambda}\brl \tilde{A}\,,\tilde{V}\brr_{\mathsf{J}^{\dagger}}^{\sim}
-\frac{1}{\lambda^{2}}\left(\partial_{\tilde{\boldsymbol{p}}}\tilde{A}\right)\overset{d}{\cdot}\partial_{\tilde{\boldsymbol{p}}}
\left(\tilde{V}-\tilde{A}\right)
=0
\end{eqnarray}
\begin{eqnarray}
\label{multi:entropy}
\hspace{-1.2cm}
\tau\,\partial_{t}\tilde{S}
+\frac{1}{\lambda}\brl \tilde{A}\,,\tilde{S}\brr_{\mathsf{J}^{\dagger}}^{\sim}
+\frac{1}{\lambda^{2}\,\beta}\tilde{\mathfrak{S}}\tilde{A}=0
\end{eqnarray}
\end{subequations}
where we used the notation $\tilde{\boldsymbol{x}}=\boldsymbol{q}\oplus\boldsymbol{\tilde{p}}$
\begin{eqnarray}
\label{}
\brl \tilde{S}\,,\tilde{V}\brr_{\mathsf{J}^{\dagger}}^{\sim}
\equiv (\partial_{\boldsymbol{\tilde{x}}}\tilde{S})\cdot
\left(\mathsf{J}^{\dagger}\cdot\partial_{\boldsymbol{\tilde{x}}}\tilde{V}\right)
=(\partial_{\tilde{\boldsymbol{p}}}\tilde{S})\overset{d}{\cdot}\partial_{\boldsymbol{q}}\tilde{V}
-(\partial_{\boldsymbol{q}}\tilde{S})\overset{d}{\cdot}\partial_{\tilde{\boldsymbol{p}}}\tilde{V}
\end{eqnarray}
and
\begin{eqnarray}
\label{multi:WiLa}
\tilde{\mathfrak{S}}:=-(\partial_{\tilde{\boldsymbol{p}}}\tilde{S})\overset{d}{\cdot}\partial_{\tilde{\boldsymbol{p}}}
+\id_{d}:\partial_{\tilde{\boldsymbol{p}}}\otimes\partial_{\tilde{\boldsymbol{p}}}
\end{eqnarray}
In what follows, we will also write $\tilde{\mathfrak{S}}^{(0)}$ to denote the replacement 
in (\ref{multi:WiLa}) of $\tilde{S}$ with its zeroth order approximation 
$\tilde{S}^{(0)}$.

As often occurs for homogenization of parabolic equation \cite{PaSt08}, 
we need to analyze the first three orders of the regular expansion in powers of
$\lambda$ in order to fully determine the leading order contributions to 
$S$ and $A$. This is because the first order is needed to assess the \emph{centering condition}
coupling the widely separated scales which we wish to resolve
in the asymptotics. The second order approximation uses the information conveyed by the 
centering condition to determine the \emph{cell problem}, a closed equation for the
effective dynamics in the limit of vanishing $\lambda$.

\subsection{Zeroth order}
\label{sec:multi:zero}

From (\ref{multi:extremal}) we get the condition
\begin{eqnarray}
\label{}
\tilde{\mathfrak{S}}^{(0)}\left(2\,A_{(0)}-V_{(0)}\right)=0
\end{eqnarray}
stating that at leading order the value function $V_{(0)}$ may differ from 
the non equilibrium Helmholtz energy at most by a function independent 
of momentum variables:
\begin{eqnarray}
\label{zero:v}
V_{(0)}=2\,A_{(0)}+V_{(0:0)}
\end{eqnarray}
where
\begin{eqnarray}
\label{}
\partial_{\tilde{\boldsymbol{p}}}V_{(0:0)}=0
\end{eqnarray} 
Once we inserted (\ref{zero:v}) in (\ref{multi:value}), (\ref{multi:entropy}) we get 
into
\begin{subequations}
\label{zero:MAK}
\begin{eqnarray}
\label{multi:zero1}
\partial_{\tilde{\boldsymbol{p}}}A_{(0)}\overset{d}{\cdot}\partial_{\tilde{\boldsymbol{p}}}A_{(0)}=0
\end{eqnarray}
\begin{eqnarray}
\label{multi:zero2}
\frac{1}{\beta}\tilde{\mathfrak{S}}^{(0)}A_{(0)}=0
\end{eqnarray}
\end{subequations}
The boundary conditions (\ref{multi:ini}), (\ref{multi:fin}) translate 
into
\begin{eqnarray}
\label{zero:neareq}
S_{\mathrm{o}}=S_{\mathrm{f}}+o(\lambda^{2})
\end{eqnarray}
Hence, we see from (\ref{zero:MAK}) that (\ref{zero:neareq}) is 
satisfied upon setting
\begin{eqnarray}
\label{zero:pdf}
S_{(0)}(\tilde{\boldsymbol{p}})=\frac{\parallel\tilde{\boldsymbol{p}}\parallel_{\id_{d}}^{2}}{2}
+S_{(0:0)}(\boldsymbol{q},t,\dots)
\end{eqnarray}
and
\begin{eqnarray}
\label{zero:neareq2}
\partial_{\tilde{\boldsymbol{p}}}A_{(0)}=\partial_{\tilde{\boldsymbol{p}}}V_{(0)}=0
\end{eqnarray}

\subsection{First order: centering condition}
\label{sec:multi:first}

Maxwell momentum distribution is the unique element of the kernel of $\tilde{\mathfrak{S}}^{(0)\dagger}$ in $\mathbb{L}^{2}(\mathbb{R}^{d})$. 
By Fredholm alternative (see e.g. \cite{PaSt08}) 
\begin{eqnarray}
\label{first:extremal}
\brl S_{(0)}\,,V_{(0)}\brr_{\mathsf{J}^{\dagger}}^{\sim}
-\frac{1}{\beta}\tilde{\mathfrak{S}}^{(0)}(2\,A_{(1)}-V_{(1)})=0
\end{eqnarray} 
admits a unique solution if and only if the solvability condition
\begin{eqnarray}
\label{first:solvability}
\lefteqn{
0=\int_{\mathbb{R}^{d}}\mathrm{d}^{d}p\,e^{-S_{(0)}}\brl S_{(0)}\,,V_{(0)}\brr_{\mathsf{J}^{\dagger}}^{\sim}=
}
\nonumber\\&&
-\int_{\mathbb{R}^{d}}\mathrm{d}^{d}p\,\brl e^{-S_{(0)}}\,,V_{(0)}\brr_{\mathsf{J}^{\dagger}}^{\sim}
=\partial_{\boldsymbol{q}}\int_{\mathbb{R}^{d}}\mathrm{d}^{d}p\,e^{-S_{(0)}}\partial_{\tilde{\boldsymbol{p}}}V_{(0)}
\end{eqnarray}
holds true which is always the case if (\ref{zero:neareq2})
is verified. Hence we conclude
\begin{eqnarray}
\label{first:v1}
V_{(1)}
=2\,\left(A_{(1)}+\tilde{\boldsymbol{p}}\cdot\partial_{\boldsymbol{q}}A_{(0)}\right)
+\tilde{\boldsymbol{p}}\cdot\partial_{\boldsymbol{q}}V_{(0:0)}+V_{(1:0)}
\end{eqnarray}
with
\begin{eqnarray}
\label{}
\partial_{\tilde{\boldsymbol{p}}}V_{(1:0)}=0
\end{eqnarray}
Turning to the value function equation (\ref{multi:value}), we see that 
\begin{eqnarray}
\label{}
\hspace{-1.0cm}
-\sum_{i=0}^{1}\partial_{\tilde{\boldsymbol{p}}}V_{(1-i)}\overset{d}{\cdot}\partial_{\tilde{\boldsymbol{p}}}A_{(i)}
+\brl A_{(0)}\,,V_{(0)}\brr_{\mathsf{J}^{\dagger}}^{\sim}
+2\,\partial_{\tilde{\boldsymbol{p}}}A_{(1)}
\overset{d}{\cdot}\partial_{\tilde{\boldsymbol{p}}}A_{(0)}=0
\end{eqnarray}
is also satisfied by (\ref{zero:neareq2}). New information
comes from the expansion of the microscopic entropy equation  
\begin{eqnarray}
\label{}
-\partial_{\boldsymbol{q}}A_{(0)}\cdot\partial_{\tilde{\boldsymbol{p}}}S_{(0)}+\tilde{\mathfrak{S}}^{(0)}A_{(1)}=0
\end{eqnarray}
which yields the \emph{centering condition} of the expansion: 
\begin{eqnarray}
\label{first:centering}
A_{(1)}=-\tilde{\boldsymbol{p}}\cdot\partial_{\boldsymbol{q}}A_{(0)}
\end{eqnarray}
The centering condition couples here a linear asymptotic behavior in $\boldsymbol{p}$ with
a non-trivial \emph{cell problem} in $(\boldsymbol{q},t)$ which we will determine by 
requiring solvability in the sense of Fredholm's alternative at order $O(\lambda^{2})$. 
Contrasting (\ref{first:centering}) with (\ref{first:v1}) we infer that 
\begin{eqnarray}
\label{}
\partial_{\tilde{\boldsymbol{p}}}V_{(1)}=\partial_{\boldsymbol{q}}V_{(0:0)}
\end{eqnarray}
It is worth here to emphasize the relevant simplification induced by the
over-damped limit. The over-damped limit entitled us to neglect the Poisson bracket 
also in the sub-leading order (\ref{first:extremal}) of the expansion
of (\ref{Occam:extremal}). The crucial consequence is that the relation 
between $V$ and $A$ remains \emph{local} within accuracy. Intermediate asymptotics 
around (\ref{zero:MAK}) do not enjoy this property. This is not surprising in light 
of example of section~\ref{sec:ex:sq} showing that, even in the Gaussian case,
the coincidence of the initial and final marginal momentum distribution does not 
imply in general thermalization.

\subsection{Second order: cell problem}
\label{sec:multi:second}

The extremal equation (\ref{multi:extremal}) yields now the condition
\begin{eqnarray}
\label{}
\hspace{-1.0cm}
\partial_{\tilde{\boldsymbol{p}}}S_{(0)}\overset{d}{\cdot}\partial_{\boldsymbol{q}}V_{(1)}
-\partial_{\boldsymbol{q}}S_{(0)}\overset{d}{\cdot}\partial_{\boldsymbol{q}}V_{(0,0)}
-\frac{1}{\beta}\tilde{\mathfrak{S}}^{(0)}(2\,A_{(2)}-V_{(2)})=0
\end{eqnarray}
The solvability condition imposes
\begin{eqnarray}
\label{second:vcomp}
\partial_{\boldsymbol{q}}V_{(0,0)}=0
\end{eqnarray}
whence
\begin{eqnarray}
\label{}
V_{(2)}=2\,A_{(2)}+\tilde{\boldsymbol{p}}\cdot\partial_{\boldsymbol{q}}V_{(1:0)}+V_{(2:0)}
\end{eqnarray}
with $\partial_{\tilde{\boldsymbol{p}}}V_{(2:0)}=0$. Hence, combining (\ref{zero:v}) with (\ref{second:vcomp}), 
(\ref{first:centering}) and (\ref{zero:neareq}), the value function equation reduces to
\begin{eqnarray}
\label{second:A}
2\,\tau\,\partial_{t}A_{(0)}
-\partial_{\boldsymbol{q}}A_{(0)}
\overset{d}{\cdot}\partial_{\boldsymbol{q}}A_{(0)}=0
\end{eqnarray}
Finally, the equation for the microscopic entropy is
\begin{eqnarray}
\label{}
\hspace{-1.6cm}
\partial_{t}S_{(0)}
-\partial_{\boldsymbol{q}}A_{(0)}\overset{d}{\cdot}\partial_{\boldsymbol{q}}S_{(0)}
+\tilde{\boldsymbol{p}}\cdot\partial_{\boldsymbol{q}}\otimes\partial_{\boldsymbol{q}}A_{(0)}\cdot\partial_{\tilde{\boldsymbol{p}}}S_{(0)}
+\frac{1}{\beta}\tilde{\mathfrak{S}}^{(0)}A_{(2)}=0
\end{eqnarray}
Regarding this latter as an equation for $A_{(2)}$ and invoking again Fredholm's alternative,
we see that it admits a unique solution if and only if 
\begin{eqnarray}
\label{second:S}
\tau\,\partial_{t}S_{(0:0)}-\partial_{\boldsymbol{q}}A_{(0)}\overset{d}{\cdot}\partial_{\boldsymbol{q}}S_{(0:0)}
+\frac{1}{\beta}\id_{d}:\partial_{\boldsymbol{q}}\otimes\partial_{\boldsymbol{q}}A_{(0)}=0
\end{eqnarray}
holds true. The system formed by the equalities (\ref{zero:pdf}), (\ref{first:centering}) and the 
\emph{cell problem} equations (\ref{second:A}), (\ref{second:S}) fully specifies the homogenization 
asymptotics we set out to derive. 

\subsection{Asymptotic expression of the entropy production}

The cell problem equations (\ref{second:A}), (\ref{second:S}) specify
a Monge--Amp\`ere--Kantorovich evolution \cite{Villani} 
between a initial configuration space state with density
\begin{eqnarray}
\label{multi:odini}
\tilde{\mathtt{m}}_{\mathrm{o}}(\boldsymbol{q})=\left(\frac{\beta}{2\,\pi}\right)^{d/2}\,e^{-\beta\,U_{\mathrm{o}}(\boldsymbol{q})}
\end{eqnarray}
and a final one with density
\begin{eqnarray}
\label{multi:odfin}
\tilde{\mathtt{m}}_{\mathrm{f}}(\boldsymbol{q})=\left(\frac{\beta}{2\,\pi}\right)^{d/2}\,\,e^{-\beta\,U_{\mathrm{f}}(\boldsymbol{q})}
\end{eqnarray}
The recovery of the  Monge--Amp\`ere--Kantorovich equations unveils the link 
between the minimum entropy production by the phase space process (\ref{def:process}) 
and the optimal control of the corresponding thermodynamic quantity which can 
be directly defined in the over-damped limit. 
As a matter of fact, the expansion of $A$ starts with the $O(\lambda)$ term specified by the 
centering condition (\ref{first:centering}) which is linear in $\tilde{\boldsymbol{p}}=\boldsymbol{p}/\lambda$.
The upshot is that the over-damped expansion of the minimum over $\mathbb{A}$ of the Langevin--Kramers entropy 
production starts with
\begin{eqnarray}
\label{multi:odentropy}
\mathcal{E}_{\tf,0}=\beta\,
\int_{0}^{\tf}\frac{\mathrm{d}t}{\tau}
\int_{\mathbb{R}^{d}}\mathrm{d}^{d}q\,\beta^{d/2}e^{-S_{(0,0)}}\left(\partial_{\boldsymbol{q}}A_{(0)}\right)
\overset{d}{\cdot}\partial_{\boldsymbol{q}}A_{(0)}+O(\lambda)
\end{eqnarray}
We therefore proved that the leading order of the expansion coincides with the minimal
entropy production by the Langevin--Smoluchowski dynamics.

\section{Discussion}
\label{sec:final}

Many physical systems are modeled by kinetic-plus-potential Hamiltonians
\begin{eqnarray}
\label{final:k+p}
H(\boldsymbol{p},\boldsymbol{q})=\frac{\parallel\boldsymbol{p}\parallel_{\id_{d}}^{2}}{2}+U(\boldsymbol{q},t)
\end{eqnarray}
The example of section \ref{sec:ex:azero} evinces that requiring (\ref{final:k+p}) adds an optimization 
constraint which is not generically satisfied by the extremal equations (\ref{Occam:eqs}) over $\mathbb{A}$. Furthermore, 
the kinetic-plus-potential hypothesis deeply affects the control problem by introducing two new difficulties. 
First, it restricts to the gradient $\partial_{\boldsymbol{q}}U$ of the potential energy the available $d$ control degrees of freedom. In this regard, 
it is worth emphasizing that it is a non-trivial consequence of H\"ormander theorem (see e.g. \cite{Rey06} and references 
therein) that a sufficiently regular (\ref{final:k+p}) is enough to generate a Fokker--Planck evolution of a smooth initial 
density for a Langevin--Kramers dynamics with degenerate noise acting only on $d$ out of $2\,d$ degrees of freedom.
Physical intuition suggests, however, that the surmise (\ref{final:k+p}) should not create an insurmountable difficulty
for controllability by which we mean the existence of a non-empty set of potentials $U(\boldsymbol{q},t)$ able to steer
a transition between two probability densities verifying physically plausible assumptions.
The second and more substantial difficulty is that inserting (\ref{final:k+p}) into (\ref{tf:entropy}) yields an
entropy production expression which depends upon the control only implicitly through the probability measure. 
Controls are in such a case only subject to the constraint imposed by the requirement of steering a finite-time 
transition between smooth probability densities. General considerations \cite{FlemingSoner} lead us to envisage
that entropy production may only attain an infimum when evaluated according to a \emph{singular control} strategy.
Such a strategy may take the form of a potential $U$ confining the momentum process within a ``inactivity region'' 
where $U$ vanishes. We expect the boundary of such inactivity region to be marked by the the vanishing of the
momentum gradient $\partial_{\boldsymbol{p}}V$ of the value function of the corresponding dynamic programming equation. 
Proving the realizability and optimality of such a control strategy are challenges lying beyond the scopes of
the present work. By insisting that the Hamiltonian belongs to the class of admissible controls $\mathbb{A}$, we 
focused instead on control strategies which we interpret as ``macroscopic'' in view the regularity assumptions on the control 
Hamiltonian. These assumptions are analogous to those adopted in previous studies of the entropy production by 
Langevin--Smoluchowski dynamics \cite{AuGaMeMoMG12} or by Markov jump processes \cite{MGMePe12}. We therefore 
gather that the existence of the entropy production minimum (\ref{Occam:eqs}),  degenerate because of non-coercivity, and 
which recovers in the over-damped limit the Monge--Amp\`ere--Kantorovich evolution, yields a robust general picture
of the ``optimal'' thermodynamics for a large class of physical processes described by Markovian evolution equations.

\section*{Acknowledgements}

It is a pleasure to thank Carlos Mej\'ia--Monasterio for discussions and useful comments on this manuscript.
The work of PMG is supported by by  the  Center  of
Excellence  ``Analysis and  Dynamics''of the  Academy of  Finland. The results of this paper were first presented
during the conference ``\emph{6-th Paladin memorial: Large deviations and rare events in physics and biology}'' Rome, 
September 23-25, 2013. The author wishes to warmly thank the organizers to give him the opportunity to partake the 
event as invited speaker.

\appendix

\section*{Appendices}

\section{Mean derivatives and current velocity of a diffusion process}
\label{ap:derivatives}

We recall that the drift of an $\mathbb{R}^{d}$-valued diffusion processes 
$\zeta\equiv\left\{\boldsymbol{\zeta}_{t}\,,\hspace{0.1cm}t\in\,[\ti\,,\tf]\right\}$ with generator
\begin{eqnarray}
\label{ap:generator}
\mathfrak{L}=\boldsymbol{b}\cdot\partial_{\boldsymbol{x}}
+\frac{1}{2}\mathsf{K}:\partial_{\boldsymbol{x}}\otimes\partial_{\boldsymbol{x}}
\end{eqnarray}
can be regarded as the \emph{mean forward derivative} of the process:
\begin{eqnarray}
\label{}
\mathrm{D}_{\boldsymbol{x}}\boldsymbol{\zeta}_{t}\equiv\lim_{\mathrm{d}t\downarrow\,0}
\mathrm{E}_{\boldsymbol{\zeta}_{t}=\boldsymbol{x}}\frac{\boldsymbol{\zeta}_{t+\mathrm{d}t}-\boldsymbol{\zeta}_{t}}{\mathrm{d}t}
=\mathfrak{L}\,\boldsymbol{x}
\end{eqnarray}
Under standard regularity hypotheses \cite{Nelson01}, it is possible 
to define the \emph{mean backward derivative} of the very same process
as
\begin{eqnarray}
\label{ap:mbd}
\mathrm{D}_{\boldsymbol{x}}^{-}\boldsymbol{\zeta}_{t}\equiv\lim_{\mathrm{d}t\downarrow\,0}
\mathrm{E}_{\boldsymbol{\zeta}_{t}=\boldsymbol{x}}\frac{\boldsymbol{\zeta}_{t}-\boldsymbol{\zeta}_{t-\mathrm{\mathrm{d}t}}}{\mathrm{d}t}
\end{eqnarray}
\begin{proposition}
Let $\zeta\equiv\left\{\boldsymbol{\zeta}_{t}\,,\hspace{0.1cm}t\in\,[\ti\,,\tf]\right\}$ be
a smooth diffusion with generator (\ref{ap:generator}) and density $\mathtt{m}$. The 
mean forward derivative is
\begin{eqnarray}
\label{}
\mathrm{D}_{\boldsymbol{x}}^{-}\boldsymbol{\zeta}_{t}=-\frac{1}{\mathtt{m}(\boldsymbol{x},t)}
\frac{\left(\mathfrak{L}^{\dagger}\,\boldsymbol{x}
-\boldsymbol{x}\,\mathfrak{L}^{\dagger}\,\right)}{\tau}\mathtt{m}(\boldsymbol{x},t)
\end{eqnarray}
\end{proposition}
\begin{proof}~
By hypothesis $\zeta$ is Markovian with density $\mathtt{m}$ in the time interval $[\ti,\tf]$. 
Given its forward transition probability density $\mathtt{p}$
the backward transition probability $\mathtt{p}_{*}$ of the same process density must then satisfy
 \begin{eqnarray}
\label{derivatives:tpdf}
\mathtt{p}_{*}(\boldsymbol{x}_{1},t_{1}|\boldsymbol{x}_{2},t_{2})=
\frac{1}{\mathtt{m}(\boldsymbol{x}_{2},t_{2})}
\mathtt{p}(\boldsymbol{x}_{2},t_{2}|\boldsymbol{x}_{1},t_{1})\,\mathtt{m}(\boldsymbol{x}_{1},t_{1})
\end{eqnarray}
for any $\boldsymbol{x}_{1},\boldsymbol{x}_{2}\in\mathbb{R}^{d}$, $t_{1},t_{2}\in [\ti,\tf]$ such that
$t_{2}\geq t_{1}$. By (\ref{derivatives:tpdf}) it follows immediately
\begin{eqnarray}
\label{}
\mathrm{E}_{\boldsymbol{\zeta}_{t}=\boldsymbol{x}}\boldsymbol{\zeta}_{t-\mathrm{d}t}=
\int\mathrm{d}^{2d}x_{1}\,\boldsymbol{x}_{1}\,\frac{\mathtt{p}(\boldsymbol{x},t|\boldsymbol{x}_{1},t-\mathrm{d}t)
\,\mathtt{m}(\boldsymbol{x}_{1},t-\mathrm{d}t)}{\mathtt{m}(\boldsymbol{x},t)}
\end{eqnarray}
If we integrate the Fokker-Planck and its adjoint equation over
a time horizon of order $O(\mathrm{d}t)$ we arrive at
\begin{eqnarray}
\label{}
\hspace{-1.4cm}\lefteqn{
\mathrm{E}_{\boldsymbol{\zeta}_{t}=\boldsymbol{x}}\boldsymbol{\zeta}_{t-\mathrm{d}t}=\boldsymbol{x}+
}
\nonumber\\&&
\hspace{-1.2cm}
\frac{\dt}{\tau}\int\mathrm{d}^{2d}x_{1}\,\boldsymbol{x}_{1}\,
\frac{\mathtt{m}(\boldsymbol{x}_{1},t)\,\mathfrak{L}\,\delta^{(2d)}(\boldsymbol{x}_{1}-\boldsymbol{x})
-\delta^{(2d)}(\boldsymbol{x}_{1}-\boldsymbol{x})\,\mathfrak{L}^{\dagger}\,\mathtt{m}(\boldsymbol{x}_{1},t)}
{\mathtt{m}(\boldsymbol{x},t)}
+O\left(\frac{\dt}{\tau}\right)
\end{eqnarray}
which inserted in the definition (\ref{ap:mbd}) yields the claim.
\end{proof}
The mean backward drift governs the Fokker-Planck evolution of the density of 
the process from $\tf$ to $\ti$ \cite{Nelson01}.
By H\"ormander theorem \cite{Rey06}, the proposition above encompasses the degenerate 
noise case described by (\ref{def:process}). We are therefore entitled to write 
\begin{eqnarray}
\label{}
\tau\,\mathrm{D}_{\boldsymbol{x}}^{-}\boldsymbol{\chi}_{t}=\mathsf{J}\cdot\partial_{\boldsymbol{x}}H
-\mathsf{G}\cdot\partial_{\boldsymbol{x}}
\left(H+\frac{2}{\beta}\ln\frac{\mathtt{m}}{\beta^{d}}\right)
\end{eqnarray}
The \emph{current velocity} of a smooth diffusion is defined as
\begin{eqnarray}
\label{ap:cv}
\boldsymbol{v}(\boldsymbol{x},t)\equiv\,\tau\,
\frac{D_{\boldsymbol{x}}+D_{\boldsymbol{x}}^{-}}{2}\boldsymbol{\zeta}_{t}
\end{eqnarray}
whence (\ref{tf:cv}) follows immediately. The advantage of the current velocity
representation is that the Fokker-Planck equation for the probability density
$\mathtt{m}$ in $[\ti\,,\tf]$ is mapped by (\ref{ap:cv}) into the deterministic
mass conservation equation 
\begin{eqnarray}
\label{}
\tau\,\partial_{t}\mathtt{m}+\partial_{\boldsymbol{x}}\cdot\boldsymbol{v}\,\mathtt{m}
=0
\end{eqnarray}

\section{Variations of the Pontryagin functional}
\label{ap:var}

We avail us of the identity (\ref{Pontryagin:generator}) to treat (\ref{Pontryagin:action})
as a functional of the independent fields $A$ and $\mathtt{m}$.
The variation of (\ref{Pontryagin:action}) with respect to the costate function being trivial, 
we restrict here the attention only to those with respect to the probability density $\mathtt{m}$
and the non-equilibrium Helmholtz energy density $A$. The boundary terms generated by the
variation of $\mathtt{m}$ vanish because of the boundary conditions (\ref{Pontryagin:bc}):
\begin{eqnarray}
\label{var:pdf}
\lefteqn{
\mathcal{A}_{\mathtt{m}}^{\prime}(\mathtt{m},V,A)=
}
\nonumber\\&&
\int_{\ti}^{\tf}\frac{\mathrm{d}t}{\tau}\,\int_{\mathbb{R}^{2\,d}}
\mathrm{d}^{2\,d}x\,\mathtt{m}^{\prime}
\left\{
\parallel\partial_{\boldsymbol{x}}A\parallel_{\mathsf{G}}^{2}
+[\tau\,\partial_{t}+(\partial_{\boldsymbol{x}}A)\cdot(\mathsf{J}^{\dagger}-\mathsf{G})\cdot\partial_{\boldsymbol{x}}]V
\right\}
\end{eqnarray}
Upon applying the definition of the brackets (\ref{def:Poisson}), (\ref{def:pseudomb}) we
arrive at (\ref{Occam:value}).
The variation of $A$ can be couched into the form
\begin{eqnarray}
\label{var:fe}
\hspace{-1.0cm}
\mathcal{A}_{A}^{\prime}(\mathtt{m},V,A)=
-\int_{\ti}^{\tf}\frac{\mathrm{d}t}{\tau}\,\int_{\mathbb{R}^{2\,d}}
\mathrm{d}^{2\,d}x\,A^{\prime}\partial_{\boldsymbol{x}}\cdot\mathtt{m}\,\left\{
2\,\mathsf{G}\cdot\partial_{\boldsymbol{x}}A
-(\mathsf{J}^{\dagger}-\mathsf{G})\cdot\partial_{\boldsymbol{x}}V
\right\}
\end{eqnarray}
Recalling the definition of the microscopic entropy (\ref{def:Mentropy}),
we see that  stationarity of (\ref{var:fe}) admits a geometric interpretation
on the De Rahm--Witten complex over $\mathbb{L}^{(2)}(\mathbb{R}^{2d},\mathtt{m}\,\mathrm{d}^{2d}x)$ \cite{Helffer} equipped with the 
exterior derivative
\begin{eqnarray}
\label{var:exterior}
\mathrm{d}_{S}=e^{-S}\mathrm{d}\,e^{S}
\end{eqnarray} 
Namely it states that the dual $\mathrm{d}_{S}^{*}$ to (\ref{var:exterior}) must 
annihilate the $\mathrm{1}$-form 
\begin{eqnarray}
\label{}
\mathsf{\alpha}=
[2\,\partial_{\boldsymbol{x}}A
+(\mathsf{J}+\mathsf{G})\cdot\partial_{\boldsymbol{x}}V]\cdot\mathrm{d}x 
\end{eqnarray}
In terms of the operator (\ref{var:WiLap}) the condition translates into 
(\ref{Occam:extremal}). We also notice that that (\ref{var:WiLap}) is a degenerate
``Witten'' Laplacian \cite{Helffer} on the same complex in consequence of the inequality 
\begin{eqnarray}
\label{Wi:negdef}
\hspace{-1.0cm}
\int_{\mathbb{R}^{2d}}\mathrm{d}^{2d}x\,\mathtt{m}\,f\,\mathfrak{S}\,f=
-\int_{\mathbb{R}^{2d}}\mathrm{d}^{2d}x\,\mathtt{m}\,\parallel\partial_{\boldsymbol{x}}f\parallel_{\mathsf{G}}^{2}\,\leq\, 0
\end{eqnarray}
holding for any $f\,\in\,\mathbb{L}^{2}(\mathbb{R}^{2d},\mathtt{m}\,\mathrm{d}^{2d}x)$. 

We end this this appendix with a remark. If the nullspace 
in $\mathbb{L}^{(2)}(\mathbb{R}^{2d},\mathtt{m}\,\mathrm{d}^{2d}x)$ of the Witten Laplacian 
\begin{eqnarray}
\label{}
\bar{\mathfrak{S}}=-(\partial_{\boldsymbol{x}}S)\cdot\partial_{\boldsymbol{x}}
+\id_{2d}:\partial_{\boldsymbol{x}}\otimes\partial_{\boldsymbol{x}} 
\end{eqnarray}
consists only of constant functions then on the De Rahm--Witten complex 
(\ref{var:exterior}) then current velocity (\ref{tf:cv}) admits the Hodge 
decomposition
\begin{eqnarray}
\label{ap:Hodge}
\boldsymbol{v}=-\partial_{\boldsymbol{x}}H_{0}+\boldsymbol{h}_{2}
\end{eqnarray}
where $H_{0}$ is a differentiable phase-space function specified by the solution of
\begin{eqnarray}
\label{Wi:H0}
\bar{\mathfrak{S}}\,H_{0}=-\brl S\,,A\brr_{\mathsf{J}^{\dagger}}+\mathfrak{S}A
\end{eqnarray}
and
\begin{eqnarray}
\label{}
\boldsymbol{h}_{2}\equiv\,e^{S}\partial_{\boldsymbol{x}}\cdot\,e^{-S}\mathsf{H}_{2}
\end{eqnarray}
 $\mathsf{H}_{2}$ being differentiable anti-symmetric rank-two tensor. By construction the
elements of the decomposition in (\ref{ap:Hodge}) are orthogonal in $\mathbb{L}^{(2)}(\mathbb{R}^{2d},\mathtt{m}\,\mathrm{d}^{2d}x)$. 

There are two interesting consequences of (\ref{ap:Hodge}). The first is that
mass-transport equation for $\mathtt{m}$ depends only upon $H_{0}$ owing to
\begin{eqnarray}
\label{ap:indep}
\partial_{\boldsymbol{x}}\cdot \left(\mathtt{m}\,e^{S}\partial_{\boldsymbol{x}}\cdot\,
e^{-S}\mathsf{H}_{2}\right)
=\beta^{d}\partial_{\boldsymbol{x}}\otimes\partial_{\boldsymbol{x}}: e^{-S}\mathsf{H}_{2}=0
\end{eqnarray}
The second is that identifying the the gradient in (\ref{ap:Hodge}) as the dissipative
component of the dynamics allows us to define the ``entropy production''
\begin{eqnarray}
\label{ap:ep}
\tilde{\mathcal{E}}_{\hoz}=\beta\,\int_{\ti}^{\tf}\frac{\dt}{\tau}
\int_{\mathbb{R}^{2d}}\mathrm{d}^{2d}x\,\mathtt{m}\,\parallel\partial_{\boldsymbol{x}}H_{0}\parallel^{2}
\end{eqnarray}
At variance with (\ref{tf:entropy}), is a coercive functional of $H_{0}$ the 
optimal control whereof reduces by (\ref{ap:indep}) to that of the Langevin--Smoluchowski 
case in $\mathbb{R}^{2d}$. It must be stressed, however, that carries different physical 
information than (\ref{tf:entropy}) since this latter depends also on $\boldsymbol{h}_{2}$.

\section*{References}

\bibliography{/home/paolo/RESEARCH/BIBTEX/jabref}{}

 \end{document}